%% file: main.tex
\documentclass[11pt, letterpaper]{article}
\usepackage[utf8]{inputenc}
\usepackage{fullpage,times}
\usepackage[T1]{fontenc}

\usepackage[procnumbered,ruled,vlined,linesnumbered]{algorithm2e}

\DontPrintSemicolon
\SetKw{KwAnd}{and}
\SetProcNameSty{textsc}
\SetFuncSty{textsc}

\usepackage{mathscinet}
\usepackage{thmtools}
\usepackage{comment}
\usepackage{thm-restate}
\usepackage{amsmath,amsfonts,amsthm,amssymb,multirow}
\usepackage{times}
\usepackage{graphicx}
\usepackage{floatpag}
\usepackage{bbold}
\usepackage{enumitem}
\usepackage{subcaption} 
\usepackage{calc}
\usepackage{tikz}
\usepackage{hyperref}
\usetikzlibrary{decorations.markings}
\tikzstyle{vertex}=[circle, draw, inner sep=0pt, minimum size=4pt, fill = black]

\usepackage{graphicx}
\usepackage{tabularx}
\makeatletter
\newcommand{\multiline}[1]{%
  \begin{tabularx}{\dimexpr\linewidth-\ALG@thistlm}[t]{@{}X@{}}
    #1
  \end{tabularx}
}
\makeatother
\usepackage{mathtools}

\makeatletter
\def\BState{\State\hskip-\ALG@thistlm}
\makeatother

\newcommand{\ceil}[1]{\lceil #1 \rceil}
\newcommand{\floor}[1]{\lfloor #1 \rfloor}

\usepackage[compact]{titlesec}
\titlespacing{\section}{0pt}{3ex}{2ex}
\titlespacing{\subsection}{0pt}{2ex}{1ex}
\titlespacing{\subsubsection}{0pt}{0.5ex}{0ex}

\newtheorem{theorem}{Theorem}[section]

\newtheorem{corollary}{Corollary}[section]
\newenvironment{proofof}[1]{\emph{Proof of #1.  }}{\hfill$\Box$}
\newtheorem{definition}{Definition}[section]
\newtheorem{lemma}{Lemma}[section]
\newtheorem{claim}{Claim}[section]
\newtheorem{hypothesis}{Hypothesis}

\newtheorem{observation}{Observation}[section]

\makeatletter
\let\c@fconjecture\c@conjecture
\makeatother

\makeatletter
\let\c@fconj\c@conj
\makeatother

\usepackage[capitalize,nameinlink]{cleveref}
\crefname{theorem}{Theorem}{Theorems}
\Crefname{lemma}{Lemma}{Lemmas}
\Crefname{claim}{Claim}{Claims}
\Crefname{observation}{Observation}{Observations}
\Crefname{algorithm}{Algorithm}{Algorithms}
\Crefname{myalgctr}{Algorithm}{Algorithms}
\Crefname{challenge}{Challenge}{Challenges}
\Crefname{figure}{Figure}{Figures}
\Crefname{hypothesis}{Hypothesis}{Hypotheses}

\def \eps {\varepsilon}

\newcommand{\ignore}[1]{}

\renewcommand{\epsilon}{\varepsilon}

\title{Approximation Algorithms and Hardness for $n$-Pairs Shortest Paths and All-Nodes Shortest Cycles}

 \author{Mina Dalirrooyfard\thanks{Email: \texttt{minad@mit.edu}. Partially supported by an Akamai Fellowship and an NSF CAREER Award.}\\MIT \and  Ce Jin\thanks{Email: \texttt{cejin@mit.edu}. Partially supported by NSF Grant CCF-2129139.}\\MIT \and Virginia Vassilevska Williams\thanks{Email: \texttt{virgi@mit.edu}. Supported by an NSF CAREER Award, NSF Grant CCF-2129139, a Google Research Fellowship and a Sloan Research Fellowship.}\\MIT \and  Nicole Wein\thanks{Email: \texttt{nicole.wein@rutgers.edu}. Supported by a grant to DIMACS from the Simons Foundation (820931). Much of this work was done while the author was at MIT.} \\DIMACS }

\date{}
\begin{document}

\maketitle
\thispagestyle{empty}

\begin{abstract}
We study the approximability of two related problems on graphs with $n$ nodes and $m$ edges: \emph{$n$-Pairs Shortest Paths ($n$-PSP)}, where the goal is to find a shortest path between $O(n)$ prespecified pairs, and \emph{All Node Shortest Cycles (ANSC)}, where the goal is to find the shortest cycle passing through each node. Approximate $n$-PSP has been previously studied, mostly in the context of \emph{distance oracles}. We ask the question of whether approximate $n$-PSP can be solved faster than by using distance oracles or All Pair Shortest Paths (APSP). ANSC has also been studied previously, but only in terms of exact algorithms, rather than approximation.

We provide a thorough study of the approximability of $n$-PSP and ANSC, providing a wide array of algorithms and conditional lower bounds that trade off between running time and approximation ratio. 

A highlight of our conditional lower bounds results is that for any integer $k\ge 1$, under the combinatorial $4k$-clique hypothesis, there is no combinatorial algorithm for unweighted undirected $n$-PSP with approximation ratio better than $1+1/k$ that runs in $O(m^{2-2/(k+1)}n^{1/(k+1)-\epsilon})$ time. This nearly matches an upper bound implied by the result of Agarwal (2014).

Our algorithms use a surprisingly wide range of techniques, including techniques from the girth problem, distance oracles, approximate APSP, spanners, fault-tolerant spanners, and link-cut trees. 

A highlight of our algorithmic results is that one can solve both $n$-PSP and ANSC in $\tilde O(m+ n^{3/2+\epsilon})$ time\footnote{$\tilde{O}$ hides sub-polynomial factors.}
with approximation factor $2+\epsilon$ (and additive error that is function of $\epsilon$), for any constant $\epsilon>0$.
For $n$-PSP, our conditional lower bounds imply that this approximation ratio is nearly optimal for any subquadratic-time combinatorial algorithm. We further extend these algorithms for $n$-PSP and ANSC to obtain a time/accuracy trade-off that includes near-linear time algorithms. 

Additionally, for ANSC, for all integers $k\geq 1$, we extend the very recent almost $k$-approximation algorithm for the girth problem that works in $\tilde{O}(n^{1+1/k})$ time [Kadria et al.\ SODA'22], and obtain an almost $k$-approximation algorithm for ANSC in $\tilde{O}(mn^{1/k})$ time. 
\end{abstract}
\clearpage
\pagenumbering{arabic}
\section{Introduction}

The focus of this paper is two basic problems concerning distances in graphs: the \emph{$n$-Pairs Shortest Paths} problem and the \emph{All-Nodes Shortest Cycles} problem.\\

\noindent \emph{$n$-Pairs Shortest Paths ($n$-PSP).} Given a (weighted or unweighted, directed or undirected) graph with $n$ nodes and $m$ edges, and a set of pairs of vertices $(s_i,t_i)$ for $1\leq i\leq O(n)$, compute the distance from $s_i$ to $t_i$ for every $i$. (For ease of notation, we denote this problem $n$-PSP even though the number of pairs is not exactly $n$, rather it is $O(n)$.)\\

\noindent\emph{All-Nodes Shortest Cycles (ANSC).} Given a (weighted or unweighted, directed or undirected) graph with $n$ nodes and $m$ edges, compute for each vertex $v$, the length of the shortest cycle containing $v$, denoted $SC(v)$.\\

As we will show, these two problems are very similar in some ways and fundamentally different in other ways. We first provide some background for the $n$-PSP problem.

\paragraph{The $n$-PSP problem.}

The $n$-PSP problem was first explicitly studied in the 90s. Aingworth, Chekuri, Indyk and Motwani \cite{DBLP:journals/siamcomp/AingworthCIM99} obtained an additive 2-approximation in time $\tilde{O}(n^2)$. The other early work on this problem has been subsequently subsumed by later results for distance oracles \cite{DBLP:journals/siamcomp/AwerbuchBCP98,DBLP:journals/siamcomp/Cohen98}. 

As far as we are aware, the $n$-PSP problem has not been explicitly studied since the 90s. However, other distance-related problems have been studied in the setting where one only cares about the distances between prespecified vertex pairs, such as pairwise distance preservers, pairwise spanners, which were first studied by Coppersmith and Elkin \cite{DBLP:journals/siamdm/CoppersmithE06} and extensively studied thereafter, as well as pairwise reachability preservers \cite{DBLP:conf/soda/AbboudB18}.

Now, we will provide some motivation for studying the $n$-PSP problem. Perhaps the most classical distance problem is All-Pairs Shortest Paths (APSP). APSP can be solved in directed graphs with non-negative edge weights in time $\tilde{O}(mn)$ simply by running Dijkstra's algorithm from each vertex. For undirected unweighted graphs, APSP can be solved using matrix multiplication in time $\tilde{O}(n^{\omega})$ \cite{seidel}, where $2\le \omega<2.373$ is the matrix multiplication exponent \cite{alman2021refined}. For directed unweighted graphs, APSP can be solved in time $\tilde O(n^{2.529})$  \cite{zwick2002all} (the bound can be slightly improved by plugging in a better rectangular matrix multiplication \cite{legallurr}).
For very large graphs, these running times can be prohibitive; even just writing down the output of size $n^2$ can be too slow. 

For many applications in both theory and practice, computing \emph{all} of the distances in the graph is overkill, and instead we only care about \emph{some} of the distances (e.g. multi-source multi-sink routing \cite{DBLP:conf/networking/LiuZSF07}, many-to-many shortest paths \cite{DBLP:conf/alenex/KnoppSSSW07}, etc.). We ask a question that has been asked many times before: 
\begin{center}\emph{Can we compute some distances in a graph faster than computing all distances?}\end{center}

This question has been approached from various angles:
\begin{itemize}
\item The simplest approach to this question is perhaps to compute all distances from a single source. The Single-Source Shortest Paths (SSSP) problem can indeed be solved much faster than APSP (by Dijkstra's algorithm in $O(m+n\log n)$ time), but has the obvious drawback that all of the distances computed have the same source. 

\item Another approach towards this question, is to compute only the \emph{extremal} distances in the graph, that is, the diameter, radius, and eccentricities (the largest distance from each vertex in the graph). For these problems, there are conditional lower bounds that rule out subquadratic time \emph{exact} algorithms for sparse graphs \cite{DBLP:conf/stoc/RodittyW13}, but there has been extensive work on \emph{approximating} these parameters quickly (see e.g. \cite{DBLP:journals/corr/abs-2106-06026}).  This approach has the drawback that it only concerns extremal distances, and one might wish to compute or approximate an \emph{arbitrary} set of distances.

\item Another approach towards this question is to construct a \emph{distance oracle}, a data structure with subquadratic space that allows one to quickly query (approximate) distances. Distance oracles are designed for the setting where we wish to know some arbitrary set of distances, but we do not know a priori \emph{which} distances. Distance oracles have been extensively studied, and various trade-offs between approximation ratio and running time are known (See \cref{sec:prior} for more detailed discussion of distance oracles).
\end{itemize}

In contrast to distance oracles, we ask the question: what if we \emph{do} know a priori which distances we wish to compute? The $n$-PSP problem is precisely this problem, where we have $O(n)$ prespecified vertex pairs. We ask the question of whether we can achieve algorithms for $n$-PSP that are faster than the algorithms directly implied by known distance oracles.

In this work, we will show that this question has different answers in different regimes. For example,
\begin{itemize}
  \item In the regime of $(1+1/k)$-approximations, we show that the $n$-PSP algorithm directly implied by Agarwal's distance oracle \cite{DBLP:conf/esa/Agarwal14} has nearly optimal running time, under the combinatorial $4k$-clique hypothesis. 
  \item For $(2+\eps,\beta)$-approximation\footnote{An $(\alpha,\beta)$ approximation algorithm means that the algorithm has multiplicative error $\alpha$ and additive error $\beta$.}, we show an $n$-PSP algorithm  that runs faster than directly applying the state-of-the-art distance oracle of Chechik and Zhang \cite{cz22}. 
\end{itemize}

Now, we turn our attention to ANSC.

\paragraph{The ANSC problem.}
The ANSC problem was first studied by Yuster \cite{DBLP:journals/ipl/Yuster11}, who gave a randomized algorithm for undirected graphs with integer weights from $1$ to $M$, in time $\tilde{O}(\sqrt{M}n^{(\omega+3)/2})$. Later, Sankowski and W\polhk{e}grzycki \cite{DBLP:journals/mst/SankowskiW19}, and independently Agarwal and Ramachandran \cite{DBLP:conf/stoc/AgarwalR18}, showed that for unweighted \emph{directed} graphs there is a deterministic $\tilde{O}(n^\omega)$ time algorithm. Agarwal and Ramachandran \cite{DBLP:conf/stoc/AgarwalR18} also gave a reduction from the \emph{Replacement Paths} problem to weighted directed ANSC. In the Replacement Paths problem, we are given a graph and a shortest path $P$ between two vertices $s$ and $t$, and the goal is to find for every edge $e \in P$, a shortest path from $s$ to $t$ that avoids $e$. The reduction of \cite{DBLP:conf/stoc/AgarwalR18} increases the edge weights by a factor of $n$, and was subsequently improved to preserve the range of edge weights by Chechik and Nechushtan \cite{DBLP:conf/icalp/ChechikN20}. 

Despite this prior work on the exact version of ANSC, as far as we know, we are the first to study the approximability of ANSC.

In this work, we show various algorithmic results for the ANSC problem. For example,
\begin{itemize}
  \item We show an almost $k$-approximation algorithm for ANSC with running time comparable to the best known $k$-approximation girth algorithm.
  \item We show a $(2+\eps,\beta)$-approximation algorithm for ANSC with subquadratic running time.
\end{itemize}
\subsection{Results for $n$-PSP and ANSC implied by prior work}\label{sec:prior}

We begin with two observations that relate $n$-PSP and ANSC. The proofs of these observations are in the appendix. The first observation is a reduction from exact $n$-PSP to exact ANSC in weighted graphs. The second observation is a reduction from ANSC to $n$-PSP in directed graphs that works for the approximation setting with any finite 
approximation factor. 

\begin{restatable}{observation}{obs}
\label{obs:pd-to-ansc-undir}
A $T(n,m)$-time algorithm solving weighted undirected ANSC exactly implies a $T(n,m)$-time algorithm for solving weighted undirected $n$-PSP exactly.
\end{restatable}

\begin{restatable}{observation}{obss}
\label{obs:ansc-to-pd-dir}
A $T(n,m)$-time algorithm solving (unweighted) directed  $n$-PSP with any finite approximation factor $\alpha\ge 1$ implies a $T(n,m)$-time algorithm for solving (unweighted) directed ANSC with approximation factor $\alpha$. 
\end{restatable}

One simple way to obtain approximation algorithms for directed or undirected $n$-PSP and directed ANSC is to use known algorithms for approximate All-Pairs Shortest Paths (APSP). The running times of these algorithms, however, will always be at least $\Omega(n^2)$ due to the size of the output of APSP.
If there is an algorithm for approximate APSP in $T(n,m)$ time for directed or undirected graphs with $n$ nodes and $m$ edges, then we can approximate directed or undirected (respectively) $n$-PSP in $O(n)+T(n,m)$ time with the same accuracy by looking up the $O(n)$ input pairs in the output of the APSP algorithm. We can approximate directed ANSC in $O(n^2)+T(n,m)$ time by computing $\min_u \hat{d}(v,u)+\hat{d}(u,v)$ for every $v$, where $\hat{d}(\cdot,\cdot)$ is the distance estimate that the APSP algorithm outputs. 
The approximation guarantee for ANSC will have the same multiplicative error  as APSP, and the additive error will double. Since the output of APSP is of size $\Theta(n^2)$, the second term in these running times is dominant and so we can approximate both directed or undirected $n$-PSP and directed ANSC in $O(T(n,m))$ time. We note that it is not clear how to get an algorithm for undirected ANSC directly from APSP, as Observation \ref{obs:ansc-to-pd-dir} works only for directed graphs. Next, we outline the known algorithms for APSP. 

We have already outlined the known algorithms for exact APSP. 
As for approximation algorithms for APSP, in directed or undirected 
graphs with non-negative edge weights, for any $\epsilon>0$, Zwick \cite{zwick2002all} gave a $(1+\epsilon)$ approximation time algorithm in time $O(\frac{n^\omega}{\epsilon}\log(W))$ where $W$ is the largest edge weight. 
In the undirected setting, APSP in graphs with integer weights in $[-W,W]$ can be solved in $O(n^{\omega}\log(W))$ \cite{shoshan1999all}. For any $\epsilon>0$, \cite{bringmann2019approximating} gives a $(1+\epsilon)$ approximation algorithm in  $O(\frac{n^\omega}{\epsilon}\rm{polylog}(\frac{n}{\epsilon}))$ time, improving upon Zwick's algorithm for large weights. There are many more algorithms for approximating APSP \cite{apsp1,apsp2,apsp3,apsp4,DBLP:journals/siamcomp/AingworthCIM99}, however since the focus of our paper is on subquadratic-time algorithms we do not describe them in detail.

Another simple way to obtain approximation algorithms for undirected $n$-PSP, is to use an approximate distance oracle (DO). A DO is a data structure that allows one to query distances. The parameters of interest in a DO are preprocessing time, query time, space, and (multiplicative and additive) approximation ratio. Given a DO with preprocessing time $p(m,n)$, query time $q(m,n)$, one can obtain an algorithm for $n$-PSP with the same approximation ratio, in time $p(m,n)+O(n)\cdot q(m,n)$ simply by querying all of the input pairs. Unlike algorithms for $n$-PSP that are based on APSP, algorithms based on DOs do not have an inherent running time of $\Omega(n^2)$.

We focus on DOs with subquadratic preprocessing time. For any integer $k\ge 2$, the Thorup-Zwick DO \cite{thorup2005approximate} has preprocessing time $O(kmn^{1/k})$, size $O(n^{1+1/k})$ with query time $O(k)$ and approximation factor $(2k-1)$. We can use this DO to obtain a $(2k-1)$-approximation algorithm for weighted undirected $n$-PSP in time $O(kmn^{1/k})$. Patrascu et al \cite{patrascu2012new} also extends the Thorup-Zwick DO to fractional values of $k$. 
Additionally, \cite{patrascu2010distance} gives a DO with preprocessing time $O(mn^{2/3})$ and constant query time that returns a path of length at most $2d+1$ when queried for a pair at distance $d$. This gives us a $(2,1)$-approximation algorithm in $O(mn^{2/3})$ time for $n$-PSP. On the other side of the time/accuracy trade-off, Agarwal \cite{DBLP:conf/esa/Agarwal14} gives a DO that yields for any integer $k\geq 1$, a $(1+\frac{1}{k})$-approximation algorithm for $n$-PSP in time $\tilde{O}(m^{2-2/(k+1)}n^{1/(k+1)})$, as well as a $(1+\frac{1}{k+0.5})$-approximation algorithm for $n$-PSP in time $\tilde{O}(m^{2-3/(k+2)}n^{2/(k+2)})$ (see \cref{sec:agrawal} for explanation). 
Additionally, there are DOs with processing times that have additive dependence between $n$ and $m$. For any integer $k\ge 1$, Wulff-Nilsen \cite{apsp1} gives a $(2k-1)$ approximate distance oracle with preprocessing time $O(\sqrt{k}m+n^{1+c/\sqrt{k}})$ for a constant $c=9+3\sqrt{13}$, and query time $O(k)$. This gives a $(2k-1)$ approximation algorithm for $n$-PSP in $O(\sqrt{k}m+n^{1+c/\sqrt{k}})$ time.
Very recently, Chechik and Zhang \cite{cz22} obtained a constant query time $(2+\eps,\beta)$-approximate distance oracle with subquadratic preprocessing time $\tilde O(m + n^{5/3+\eps})$, which immediately implies a $(2+\eps,\beta)$-approximate $n$-PSP algorithm in $\tilde O(m + n^{5/3+\eps})$ time.

Note that we cannot use Observation \ref{obs:pd-to-ansc-undir} to obtain an algorithm for ANSC since it only works in the exact setting, and it is not clear how to use distance oracles in general to solve ANSC in the undirected setting. 
Moreover, there are no non-trivial distance oracles for the directed setting \cite{thorup2005approximate}.

Thus, past work doesn't give any subquadratic approximation algorithms for ANSC in directed or undirected graphs. For undirected $n$-PSP we get the above results from distance oracles, but it is not clear if this is the best one can do.

From the lower bounds side, ANSC is closely related to Girth (the problem of finding the smallest cycle in the graph), and any conditional lower bounds for Girth immediately carry over to ANSC. A previously known lower bound for the Girth in directed graphs states that
under the $k$-Cycle Hypothesis, any better than $2$-approximation for Girth requires time $m^{2-o(1)}$ \cite{girth-icalp}.

\subsection{Our Results}

We provide a thorough study of the approximability of $n$-PSP and ANSC, providing a wide array of algorithms and conditional lower bounds that trade off between running time and approximation. 

Before stating our results we provide the main hardness assumptions that we use for our conditional lower bounds.

\subsubsection{Main Hardness Assumptions}
We stress that our conditional lower bounds are based on well-established hardness assumptions in fine-grained complexity. We obtain hardness results based on a number of different assumptions, but for the sake of clarity, we only list the two central ones here. For a complete list of the assumptions we use and the associated hardness results see \cref{app:assump}.

Our first hypothesis concerns \emph{combinatorial algorithms} for $k$-clique detection, and has been used as a hardness hypothesis in \cite{DBLP:conf/innovations/Lincoln020,DBLP:journals/corr/abs-2106-15524,DBLP:journals/talg/BringmannGMW20,DBLP:conf/focs/AbboudBW15a,lincolnsoda18,bergamaschi2021new}. By ``combinatorial'' we mean algorithms that do not use the heavy machinery of Fast Matrix Multiplication. 

\begin{hypothesis}[Combinatorial $k$-Clique Hypothesis] \label{hyp:kclique}
Let $k\ge 3$ be a constant integer. In the word-RAM model with $O(\log n)$ bit words,
 there is no $O(n^{k-\eps})$ time \emph{combinatorial algorithm} for $k$-clique detection, for any constant $\eps>0$.
\end{hypothesis}

The Combinatorial $3$-Clique Hypothesis is also called \emph{Combinatorial Dense Triangle Hypothesis}, which is equivalent to the \emph{Combinatorial Boolean Matrix Multiplication Hypothesis} (see \cref{app:assump}).


Our next hypothesis is used for our main conditional lower bound result. This hypothesis was introduced in \cite{lincolnsoda18} and concerns (not necessarily combinatorial) algorithms for $k$-clique detection in hypergraphs. It has been widely used as a hardness assumption \cite{DBLP:conf/pods/CarmeliTGKR21,DBLP:journals/corr/abs-1712-07880,DBLP:journals/tods/CarmeliK21,DBLP:conf/coco/BringmannFK19,DBLP:conf/icalp/BringmannS21,DBLP:journals/corr/abs-1912-10704,DBLP:conf/coco/KunnemannM20,DBLP:conf/soda/WilliamsX20,DBLP:conf/pods/000120,DBLP:conf/stoc/AbboudBDN18}. 
\begin{hypothesis}[$(k,r)$-Hyperclique Hypothesis] \label{hyp:hyperclique}
Let $k> r\ge 3$ be a constant integer. In the word-RAM model with $O(\log n)$ bit words,
 there is no $O(n^{k-\eps})$ time algorithm for $k$-clique detection in $r$-uniform hypergraphs, for any constant $\eps>0$.
\end{hypothesis}

\subsubsection{Hardness Results}

Our hardness results for $n$-PSP are shown in \cref{table:pdist-lb}, and those for ANSC are shown in \cref{table:ansc-lb}. For the sake of clarity we only describe our most central hardness results in words here. For a detailed statement of all of our hardness results see \cref{app:assump,app:lb}.

Most of our hardness results are actually stronger than needed for Pairwise Distance and ANSC. In particular, they also apply for weaker problems, specifically the \emph{extremal} versions of $n$-PSP and ANSC where the output is either the minimum or maximum among the $O(n)$ values outputted. In particular, the \emph{girth} is the length of the smallest cycle in the graph, that is, the minimum value in the output of ANSC. We define the \emph{cycle diameter} as the largest value in the output of ANSC, that is, the maximum over all vertices $v$ of the length of the smallest cycle through $v$. Analogously for $n$-PSP, given a graph and $O(n)$ pairs of vertices $(s_i,t_i)$, we define the \emph{$n$-pairs minimum distance} as $\min_i d(s_i,t_i)$, and the \emph{$n$-pairs diameter} as $\max_id(s_i,t_i)$.

\begin{table}[h!]
\centering
\begin{tabular}{ |c|c|c|c|c| } 
 \hline
Approximation & Running time LB & Theorem & Hypothesis & Comments \\ 
 \hline\hline
 $5/3-\eps$ & $m^{1+\delta - o(1)}$ & Thm \ref{thm:tripair} & Sparse Triangle &  \\  \hline
 $5/3-\eps$ & $n^{\omega - o(1)}$ & Thm \ref{thm:tripair} & Dense Triangle & \\  \hline
  $5/3-\eps$ & $n^{\omega - o(1)}$ & Thm \ref{thm:simppair} & Simplicial Vertex & for $n$-Pairs Diameter\\  \hline
 $2-\eps$ & $n^{3/2 - o(1)}$ & Thm \ref{thm:2approx-lb-pairwise-distances} & Comb.\ BMM & \\  \hline
 $2-\eps$ & $m+n^{2-o(1)}$ & Cor \ref{cor:clique1} & Comb.\ $4$-clique & {\bf nearly matches Thm \ref{thm:pd-2approx-16time}}  \\ \hline
 $3-4/k-\eps$ & $m + n^{k/(k-2)-o(1)}$ & Cor \ref{cor:clique1} & Comb.\ $k$-clique &  \\  \hline
 $1+1/k-\eps$ & $m^{2-2/(k+1)}n^{1/(k+1)-o(1)}$ & Cor \ref{cor:clique2} & Comb.\ $4k$-clique & {\bf nearly matches \cite{DBLP:conf/esa/Agarwal14} } \\  \hline
 $1+1/(k+0.5)-\eps$ & $m^{2-3/(k+2)}n^{2/(k+2)-o(1)}$ & Cor \ref{cor:clique3} & Comb.\ $(4k+2)$-clique & {\bf nearly matches \cite{DBLP:conf/esa/Agarwal14} } \\  \hline
 $3/2-\eps$ & $mn^{1/2-o(1)}$ & Cor \ref{cor:hyper2} & $(4,3)$-hyperclique &   \\  \hline
 $3-6/k-\eps$ & $m+n^{k/(k-2)-o(1)}$ & Cor \ref{cor:hyper2} & $(k,3)$-hyperclique &   \\  \hline

 any finite & $m^{2 - o(1)}$ & Thm \ref{thm:kcycle} & $k$-cycle &  directed graphs\\  \hline
\end{tabular}
\caption{Conditional Lower Bounds for (unweighted) $n$-PSP. All results are for undirected graphs unless otherwise specified. All results work for $n$-Pairs Minimum Distance unless otherwise specified. See the body of the text for details about the near-tightness of some of our conditional lower bounds.}
\label{table:pdist-lb}
\end{table}

\begin{table}[h!]
\centering
\begin{tabular}{ |c|c|c|c|c| } 
 \hline
Approximation & Running time LB & Theorem & Hypothesis & Comments \\ 
\hline\hline
 $4/3-\eps$ & $m^{1+\delta - o(1)}$ & Obs \ref{obs:lb} & Sparse Triangle & for Girth \\  \hline
 $4/3-\eps$ & $n^{\omega - o(1)}$ & Obs \ref{obs:lb} & Dense Triangle & for Girth\\ 
  \hline
  $7/5-\eps$ & $n^{\omega - o(1)}$ & Thm \ref{thm:simpcycle} & Simplicial Vertex & for Cycle Diameter\\ 
 \hline
 $3/2-\eps$ & $m^{4/3-o(1)}$ & Thm \ref{thm:aest} & All-Edges Sparse Triangle & \\
  \hline
  $3/2-\eps$ & $n^2$ & Thm \ref{thm:an4c} & unconditional & \\

 \hline
\end{tabular}
\caption{Conditional Lower Bounds for (unweighted) undirected ANSC.}
\label{table:ansc-lb}
\end{table}

We begin with a simple hardness result for $n$-Pairs Minimum Distance.
Using a simple reduction, we get the following theorem is against a $(2-\eps)$-approximation by a combinatorial algorithm.

\begin{restatable}{theorem}{twoapproxlbpd}\label{thm:2approx-lb-pairwise-distances}
Under the Combinatorial Dense Triangle Hypothesis, any better than $2$-approximation combinatorial algorithm for $n$-Pairs Minimum Distance requires $n^{3/2-o(1)}$ time.
\end{restatable}


\cref{thm:2approx-lb-pairwise-distances} has two main drawbacks: the running time is not as high as we would like, and it is only for combinatorial algorithms. We overcome both of these drawbacks. We achieve stronger running time bounds under a generalized version of the Combinatorial Dense Triangle Hypothesis: the Combinatorial $k$-Clique Hypothesis. We also remove the ``combinatorial'' condition of \cref{thm:2approx-lb-pairwise-distances} under the $(k,r)$-Hyperclique Hypothesis. To achieve both of these goals, as well as establish a wide range of time/accuracy trade-off lower bounds for both combinatorial and non-combinatorial algorithms, some of which are nearly tight, we introduce the following general theorem. This theorem is our most technically substantial conditional lower bound. After stating the theorem, we will highlight some of its corollaries.

\begin{restatable}{theorem}{cliquereduction}
\label{lem:cliquereduction}
For integers $r,k,t$ satisfying $k-1 \ge t+1 \ge r\ge 2$, let 
\[ D =2r(t+1)-(2r-3)k,\]
and suppose $k<D$.  Then the following holds: 

 Given a $k$-(hyper)-clique instance on an $n$-vertex $r$-uniform (hyper)-graph $G$, we can reduce it (in linear time) to an unweighted undirected $n$-Pairs Minimum Distance instance with $O(kn^{t})$ vertices and $O(kn^{t+1})$ edges, such that: 
\begin{enumerate}
    \item If $G$ contains a $k$-(hyper)-clique, then the $n$-pairs minimum distance equals $k$.\label{claim1}
    \item
    If $G$ does not contain a $k$-(hyper)-clique, then the $n$-pairs minimum distance is at least $D$. 
     \label{claim2}
\end{enumerate}
\end{restatable}

Our first two corollaries of \cref{lem:cliquereduction} concern the graph (not hypergraph) version of \cref{lem:cliquereduction}, and are under the Combinatorial $k$-Clique Hypothesis. 

\cref{cor:clique1} establishes a time/accuracy trade-off against algorithms with faster running times and higher approximation ratio, while \cref{cor:clique2} establishes a time/accuracy trade-off against algorithms with better approximation ratios and slower running times. 

\begin{restatable}{corollary}{clique}
\label{cor:clique1}
For $k\ge 4$, assuming the Combinatorial $k$-Clique Hypothesis, there is no combinatorial algorithm for unweighted undirected $n$-Pairs Minimum Distance with approximation ratio better than $(3-4/k)$ in $m\cdot n^{1/(k-2) -\eps}$ time or $m +n^{k/(k-2)-\eps}$ time, for any constant $\eps>0$.
\end{restatable}

\begin{restatable}{corollary}{cliquee}
\label{cor:clique2}
For $k\ge 1$, assuming the Combinatorial $4k$-Clique Hypothesis, there is no combinatorial algorithm for unweighted undirected $n$-Pairs Minimum Distance with approximation ratio better than $(1+1/k)$ in $n^{2-\eps}$ time or $m\cdot n^{1-1/(2k)-\eps}$ time or 
$m^{2-2/(k+1)}\cdot n^{1/(k+1)-\eps}
$ time, for any constant $\eps>0$.
\end{restatable}

\begin{restatable}{corollary}{cliqueee}
\label{cor:clique3}
For $k\ge 1$, assuming the Combinatorial $(4k+2)$-Clique Hypothesis, there is no combinatorial algorithm for unweighted undirected $n$-Pairs Minimum Distance with approximation ratio better than $(1+1/(k+0.5))$ in $n^{2-\eps}$ time or $m\cdot n^{1-1/(2k+1)-\eps}$ time or 
$m^{2-3/(k+2)}\cdot n^{2/(k+2)-\eps}
$ time, for any constant $\eps>0$.
\end{restatable}

\cref{cor:clique2} is nearly tight with previously known algorithms for $n$-PSP in the following sense.
For any integer $k\geq 1$, there is a $(1+1/k)$-approximation for $n$-PSP in time $\tilde{O}(m^{2-2/(k+1)}n^{1/(k+1)})$ \cite{DBLP:conf/esa/Agarwal14} (see also \cref{sec:agrawal}), while \cref{cor:clique2} says that there is no better than $(1+1/k)$-approximation in $\tilde{O}(m^{2-2/(k+1)}n^{1/(k+1)-\eps})
$. That is, one conditionally cannot simultaneously improve both the running time and the approximation factor of the known algorithms, for any $k$.

Similarly, \cref{cor:clique3} is also nearly tight with another algorithm for $n$-PSP implied by \cite{DBLP:conf/esa/Agarwal14}'s results, which has approximation ratio $1+1/(k+0.5)$ and running time $\tilde{O}(m^{2-3/(k+2)}n^{2/(k+2)})$ (see also \cref{sec:agrawal}). \cref{cor:clique3} says that one conditionally cannot simultaneously improve both the running time and the approximation factor of this algorithm for any $k$.

Our final  corollary concerns the hypergraph version of \cref{lem:cliquereduction} and is under the $(k,r)$-Hyperclique Hypothesis (for $r\ge 3$). Unlike, the above two corollaries, the following corollary is for not necessarily combinatorial algorithms.

\begin{restatable}{corollary}{hyperr}
\label{cor:hyper2}
For $k\ge 4$, assuming the $(k,3)$-Hyperclique Hypothesis, there is no algorithm for unweighted undirected $n$-Pairs Minimum Distance with approximation ratio better than $(3-6/k)$ in
$n^{k/(k-2)-\eps}$ or $mn^{1+1/(k-2)-\eps}$ time, for any constant $\eps>0$.
\end{restatable}

Assuming the $(4,3)$-Hyperclique Hypothesis, \cref{cor:hyper2} rules out algorithms with
approximation ratio better than $3/2$ in $n^{2-\eps}$ time or $mn^{1/2-\eps}$ time, for any constant $\eps>0$. This is the choice of parameters with the best possible running time. 
As we discuss later, this is nearly tight with our algorithm from \cref{thm:pd-2approx-16time}.

By setting $k$ to be large in \cref{cor:hyper2}, we obtain lower bounds with approximation ratio larger than 2 and close to $3$. We note that this is the first known lower bound with approximation ratio higher than 2 for {\em any distance problem} except for the \emph{$ST$-Diameter} problem, but unlike in $n$-PSP, the number of vertex pairs one considers in $ST$-Diameter is much larger than the running time lower bounds.

We also note that by setting $k=6$ in \cref{cor:hyper2}, we get the same bound as \cref{thm:2approx-lb-pairwise-distances}, but for not necessarily combinatorial algorithms.

\paragraph*{Comparison with independent work \cite{cycle-removal}.} Very recently, Abboud, Bringmann, Khoury, and Zamir \cite{cycle-removal}
proposed the ``cycle-removal'' framework, and used it to obtain new conditional lower bounds for various graph problems related to approximating distances or girth. In particular, they showed super-linear lower bounds on the preprocessing time of $k$-approximate distance oracles. Their lower bounds also applied to the offline setting of distance oracle queries, which is almost the same as the $n$-PSP problem we considered here, except that they did not fix the number of query pairs to be $n$. Their results imply that, under either 3-SUM hypothesis or APSP hypothesis, for any constant $k\ge 4$, $n$-PSP does not have $k$-approximation algorithms in $m^{1+ c/k}$ time, where $c>0$ is some universal constant.
Their result has the right form $m^{1 + \Theta(1/k)}$, but it appears difficult to obtain the best possible constant $c$ on the exponent using their framework.

The main difference between their lower bound results and ours (\cref{cor:clique1,cor:clique2,cor:clique3}) is that we  focus on approximation ratio much closer to $1$, such as $1+1/k-\eps$, while they focus on arbitrarily large constant approximation ratio. \cref{cor:clique2} and \cref{cor:clique3} nearly match the known upper bounds \cite{DBLP:conf/esa/Agarwal14}, without losing constant factors on the exponent. The downside of \cref{cor:clique1,cor:clique2,cor:clique3} is that they only work against combinatorial algorithms. 

Some of our other corollaries (such as \cref{cor:hyper2}) obtained $m+n^{1+1/(k-2)-o(1)}$ lower bounds $n$-PSP with approximation ratio $3-\Theta(1/k)$. For large enough constant $k$, \cite{cycle-removal} implies a stronger lower bound, but based on different hypotheses (3-SUM or APSP) from ours (hyperclique hypothesis).

\subsubsection{Algorithmic Results}
We investigate approximation algorithms for both the $n$-PSP and ANSC problems in both directed and undirected settings in $n$-node $m$-edge graphs. Additionally, we are interested in the dependency between $m$ and $n$ in the running time of our algorithms. We first present algorithms where the running time shows a \emph{multiplicative dependency} between $n$ and $m$. Then we investigate approximation algorithms for $n$-PSP and ANSC whose running time has \emph{additive} dependence between $n$ and $m$, in particular running times of the form $m+n^{2-\eps}$.
Algorithms of this form are desirable in part because they yield near-linear time algorithms for dense enough graphs. Moreover, algorithms of the form $m+n^{2-\eps}$ have been studied for a variety of problems, for instance in distance oracles \cite{apsp1,cz22}, and recent results on bipartite matching and related problems \cite{DBLP:conf/focs/BrandLNPSS0W20}. Another motivation for studying such algorithms is that known undirected Girth algorithms do not have any multiplicative dependency on $m$, and so we ask how crucial this multiplicative dependency is for undirected ANSC and $n$-PSP. (Known \emph{directed} Girth algorithms, however do have multiplicative dependency on $m$.) 

We let an $(\alpha,\beta)$-approximation denote an approximation algorithm that outputs an estimate $\hat{x}$ for $x$ such that $x\leq \hat{x}\leq \alpha\cdot x+\beta$.

\paragraph{The $n$-PSP problem.}
Our algorithmic results for $n$-PSP are shown in \cref{table:pdist-res}. Note that due to \cref{thm:kcycle} there is no constant factor approximation algorithm for the directed case, and hence all of our algorithmic results for $n$-PSP are for the undirected case.

\begin{table}[h!]
\centering
\begin{tabular}{ |c|c|c|c| } 
 \hline
Approximation & Running time & Theorem & Comments \\ 
 \hline\hline
 $(2+\epsilon,f(\epsilon))$ & $m+n^{3/2+\eps}$ & Thm 
 \ref{thm:pd-2approx-16time}
  & $\epsilon>0$, some function $f$, {\bf nearly matches Cor \ref{cor:clique1}}\\
 \hline
 $2k-2$ & $mn^{1/k}$ & Thm \ref{thm:pd-k-tz} & integer $k\geq 2$ \\ 
 \hline
$(2k-1)\cdot (2k-2)$ & $m+n^{1+2/k}$ & Thm 
\ref{thm:pd-ksquared-2ktime}
& integer $k\geq 2$ \\
 \hline
 $2$ & $m\cdot n^{(1+\omega)/8}$ & Thm \ref{thm:st} & for $ST$-Shortest Paths, $|S|,|T|= O(\sqrt{n})$\\
 \hline
\end{tabular}
\caption{A summary of all of our approximation algorithms for $n$-PSP in undirected unweighted graphs. (Some results also work for weighted graphs as stated in each theorems.) All running times above are within polylog$(n)$ factors.} 
\label{table:pdist-res}
\end{table}

First \cref{thm:pd-k-tz} gives a straightforward nearly $2k-2$-approximation for $n$-PSP, and we leave it as an open problem whether one can achieve a $k$-approximation with similar running time.

Particularly related to our work, Chechik and Zhang \cite{cz22} obtained various distance oracles for unweighted undirected graphs with subquadratic construction time and constant query time. Their result immediately implies a $(2+\eps,\beta)$-approximate $n$-PSP algorithm in $\tilde O(m + n^{5/3+\eps})$ time. In this work we obtain a faster algorithm for $n$-PSP, stated in the following theorem.

\begin{restatable}{theorem}{almosttwopdundir}
\label{thm:pd-2approx-16time}
Given an $n$-node $m$-edge undirected unweighted graph $G$ and vertex pairs $(s_i,t_i)$ for $1\leq i\leq O(n)$, for any constant $\eps>0$, there is a randomized algorithm that computes a $(2+\eps,\beta)$-approximation for $n$-PSP in $\tilde{O}(n^{3/2+\eps}+m)$ time with high probability, for some constant $\beta$ depending on $\eps$.
\end{restatable}

\cref{thm:pd-2approx-16time} is nearly tight with our conditional lower bound from \cref{cor:clique1} in the sense that \cref{cor:clique1} conditionally rules out a $(2-\eps)$-approximation in time $n^{2-\Omega(1)}$, while \cref{thm:pd-2approx-16time} provides a $(2+\eps,f(\eps))$-approximation in time polynomially faster than $O(n^2+m)$.

\paragraph{The ANSC problem.}
Our algorithmic results for ANSC are shown in \cref{table:ansc-res}.

\begin{table}[h!]
\centering
\begin{tabular}{ |c|c|c|c| } 
 \hline
Approximation & Running time & Theorem & Comments \\ 
 \hline\hline
  $2k+1+\epsilon$ & $mn^{\alpha_k}$ & Thm \ref{thm:2k+1ansc-approx} & directed graphs, $\alpha_k$ solves $\alpha_k(1+\alpha_k)^{k-1}=1-\alpha_k$  \\ 
 \hline
   $2+\eps$ & $mn^{1/2}$ & Thm \ref{thm:2-ansc-approx} & directed graphs, {\bf nearly matches Thm 5.1 of \cite{girth-icalp}} \\ 
 \hline
 $k+\epsilon$ & $mn^{1/k}$ & Thm \ref{thm:kapprox-ansc-undir} & $\epsilon>0$ and integer $k\geq 2$ \\ 
 \hline
  nearly $1+1/(k-1)$ & $m^{2-2/k}n^{1/k}$ & Thm  
 \ref{thm:unpubansc}
 & integer $k\geq 2$\\ 
 \hline 
 $(6,1)$ & $m+n^{2-1/6}$ & Thm \ref{thm:6-1anscundir} & \\ 
 \hline 
 $(2+\epsilon,\beta)$ & $m+n^{1.5+\eps}$ & Thm \ref{thm:ansctwoapproxbest} & $\eps>0$ and $\beta$ is a function of $\eps$\\
 \hline
 $(k^2,k^3 2^{k+1})$ & $m+n^{1+2/k}$ & Thm \ref{thm:k^2-approx-ansc} & integer $k\geq 2$ \\
 \hline
\end{tabular}
\caption{A summary of all of our approximation algorithms for unweighted ANSC. (Some results also work for weighted graphs as stated in each theorem.) All running times above are within polylog$(n)$ factors. All results are for undirected graphs unless otherwise specified.
}
\label{table:ansc-res}
\end{table}

For directed graphs, we provide approximation algorithms for ANSC, showing a strong separation between ANSC and $n$-PSP in the directed case. To obtain our algorithms for directed ANSC, we generalize previously known results from Girth by Dalirrooyfard and Vassilevska W. \cite{girth-icalp}, with a slight loss in the accuracy and the running time. These results are stated in \cref{table:ansc-res} and the appendix.

We now move to the case of undirected graphs.
Recall that from Table \ref{table:pdist-res}, \cref{thm:pd-k-tz} gives a nearly $2k-2$-approximation for $n$-PSP. For ANSC, however, we are able to achieve the better approximation ratio of $(k+\eps)$, as stated in the following theorem.

\begin{restatable}{theorem}{kanscundir}
\label{thm:kapprox-ansc-undir}
Given an $n$-node $m$-edge undirected graph $G$ with edge weights in $\{1,\ldots,M\}$, a constant $\epsilon>0$ and an integer $k\ge 3$, there is a randomized algorithm that computes a $(k+\epsilon)$-approximation for ANSC in $\tilde{O}(mn^{1/k}\log{(M)})$ time with high probability.
\end{restatable}

Very recently, a result similar to \cref{thm:kapprox-ansc-undir} was shown for the Girth problem \cite{sodagirth}: an almost $k$-approximation algorithm for the undirected girth in $\tilde{O}(n^{1+1/k})$ time. Our running time for ANSC is instead $\tilde{O}(mn^{1/k})$ and this dependence on $m$ is not unexpected since for Girth, one generally only runs Dijkstra's algorithm until finding a cycle which takes $\tilde{O}(n)$ time, whereas for ANSC, we execute Dijkstra's algorithm to completion.

Now, we move to algorithms for ANSC of the form $m+n^{2-\eps}$. We consider these to be our main algorithmic results.

We begin with a ``proof of concept'' algorithm which shows that there is indeed an algorithm for ANSC with constant multiplicative and additive factors in time $m+n^{2-\eps}$ for constant $\eps$.

\begin{restatable}{theorem}{sixoneanscundir}\label{thm:6-1anscundir}
Given an $n$-node $m$-edge undirected unweighted graph $G$, there is a randomized algorithm that computes a $(6,1)$-approximation for ANSC in $\tilde{O}(m+n^{2-1/6})$ time with high probability.
\end{restatable}

We will significantly improve upon \cref{thm:6-1anscundir} in running time and multiplicative factor in our next result. However, our next algorithm does not \emph{strictly} improve upon \cref{thm:6-1anscundir} partially due to its additive error of only 1. 

Our goal is to reduce the multiplicative approximation ratio as much as possible, with the goal of getting it down to nearly 2, to match our above algorithm for $n$-PSP.

\begin{restatable}{theorem}{ansctwoapproxbest}\label{thm:ansctwoapproxbest}
Given an $n$-node $m$-edge undirected unweighted graph $G$ and a constant $\epsilon>0$, there is a randomized algorithm that computes a $(2+\epsilon,\beta)$-approximation for ANSC in $\tilde{O}(n^{1.5+\eps}+m)$ time, where $\beta$ is a constant depending only on $\epsilon$. 
\end{restatable}

We also use fast matrix multiplication to obtain improvement results for the $ST$-shortest paths problem, which is a special case of $n$-PSP. They are included in \cref{app:st}.

\section{Technical Overview}\label{sec:techniques}

We use many different techniques that were originally designed for a range of different problems and data structures, such as girth, APSP, distance oracles, spanners, fault-tolerant spanners, the simplicial vertex problem, and link-cut trees. The applicability of some of these problems to approximate $n$-PSP and ANSC is perhaps unexpected. For example, it is not clear how something like a fault-tolerant spanner would be useful in a setting that does not involve faulty vertices or edges. 

Although we pull together results from a variety of different problems, our results are not ``just'' an application of prior techniques. In the following overview of our techniques, we provide an overview of many of our results, choosing to highlight certain results that require significantly new ideas from prior work. In particular, we highlight a collection of lower bounds for $n$-PSP, as well as our collection of approximation algorithms for ANSC with running times of the form $\tilde{O}(m+n^{2-\eps})$.

\subsection{Conditional Lower Bounds}
Our conditional lower bounds are from standard hardness assumptions for basic problems such as triangle detection, $k$-cycle, and $k$-clique. Many of our conditional lower bounds are quite straightforward reductions from these problems. One of our conditional lower bounds, however, is more technically substantial, and we highlight it next.

\subsubsection{Highlight: Hardness of $n$-PSP from $k$-(Hyper)Clique Hypotheses}

For the undirected unweighted $n$-PSP problem, by adapting known techniques one could only prove fine-grained hardness for approximation ratio less than $2$.  To overcome this issue,  we give an interesting and novel reduction from the Combinatorial $k$-Clique Hypothesis to the $n$-Pairs Minimum Distance problem, which not only yields lower bounds for approximation ratio higher than 2, but also gives tight bounds that match some of our algorithms in the low-approximation regime. We believe this powerful reduction will inspire more fine-grained hardness results for related problems.

As an illustrative example, we describe the case of $k=4$. Suppose we are given a 4-clique instance on a 4-partite graph $G$ with vertex partition $V_1,V_2,V_3,V_4$ where $|V_i|=n$. In our reduction we create a 5-partite graph $G'$ with vertex partition  $V_{12},V_{23},V_{34},V_{41},V'_{12}$ from left to right, where each part contains $n^2$ vertices, and we will connect edges only between adjacent parts. 
The vertices in $V_{12}$ (and $V_{12'}$) are indexed by vertex pairs $(v_1,v_2)\in V_1\times V_2$ from the input graph $G$, and $V_{23},V_{34},V_{41}$ are similarly (according to the subscripts) indexed by vertex pairs from $G$. For every $v_1\in V_1,v_2\in V_2,v_3\in V_3$, we connect an edge between the two vertices $(v_1,v_2)\in V_{12},(v_2,v_3)\in V_{23}$ in $G'$, if and only if $(v_1,v_2,v_3)$ form a $3$-clique in $G$. We similarly add edges between $V_{23},V_{34}$, between $V_{34},V_{41}$, and between $V_{41},V'_{12}$. The resulting graph $G'$ has $N=5n^2$ vertices and at most $M=4n^3$ edges. 

If $G$ contains a 4-clique $(v_1,v_2,v_3,v_4)$, then in $G'$ there exists a length-$4$ path from $(v_1,v_2)\in V_{12}$ to $(v_1,v_2)\in V_{12}'$, by simply going from left to right along the edges specified by $v_1,v_2,v_3,v_4$.
The interesting part is the converse case: when $G$ does not contain a 4-clique, what is the shortest possible length of any path from $(v_1,v_2)\in V_{12}$ to $(v_1,v_2)\in V_{12}'$? Note that this question amounts to solving an $n$-Pairs Minimum Distance instance on $G'$. The answer cannot be $4$, since a length-4 path would immediately recover a 4-clique (containing $(v_1,v_2$) in $G$. The next smallest length is $6$ (with one backward step, and one extra forward step), and it turns out that a length-$6$ path would also have to recover a 4-clique in $G$: for example, the following length-$6$ path \[(v_1,v_2)\to (v_2,v_3)\to (v_3,v_4) \to (v_2',v_3)\to (v_3,v_4')\to (v_4',v_1)\to (v_1,v_2)\]
implies the existence of the 4-clique $(v_1,v_2,v_3,v_4')$; the other possibilities of length-$6$ paths can be similarly verified. Hence, the smallest possible answer is 8. 
Then, any better than 2-approximation combinatorial algorithm for $n$-Pairs Minimum Distance with $m\cdot n^{1/2-\eps} $ running time would be able to distinguish these two cases of distance $4$ or $\geq 8$ and hence solve the 4-clique instance, in $M\cdot N^{1/2-\eps}\le O(n^{4-\eps/2})$ time, contradicting the combinatorial 4-clique hypothesis. This lower bound nearly matches the upper bound from \cite{DBLP:conf/esa/Agarwal14}.

The reduction described above can be generalized to larger $k$, but naive ways to prove the lower bound on the shortest distance would require an exhaustive case analysis, which does not work for general values of $k$. To overcome this issue, we provide a clean combinatorial argument that can pinpoint the vertices participating in a $k$-clique when there is a too short path. This combinatorial argument is highly extendable:  It yields lower bounds with approximation ratio that increases with $k$, and it even allows us to capture \emph{hypercliques} rather than cliques. As a result, we can prove \emph{non-combinatorial} lower bounds based on the $(r,k)$-hyperclique hypothesis (albeit with slightly worse exponents compared to the combinatorial ones).

As previously noted, several parameter regimes of this conditional lower bound are nearly tight with algorithms (both previously known algorithms and new ones).

\subsection{Approximation algorithms}

\paragraph{\textsc{CycleEstimationDijkstra} data structure for ANSC.} Algorithms for approximating distances and for approximating Girth, generally have the following structure: Run Dijkstra's algorithm from a random sample of vertices, and run a \emph{truncated} version of Dijkstra's algorithm from a large set of vertices. For ANSC, we can employ a similar strategy, however the situation becomes slightly more complicated. This is because when we perform Dijkstra's algorithm for ANSC and we detect a cycle, we would like to update the estimate of $SC(v)$ for \emph{all} vertices $v$ on the cycle. To accomplish this, we employ a data structure that we call \textsc{CycleEstimationDijkstra}, which uses a modified version of Dijkstra's algorithm along with the power of \emph{link-cut trees} \cite{DBLP:journals/jcss/SleatorT83} to keep track of the relevant cycle information for all vertices.  

Our warm-up algorithm for ANSC that gives a 2-approximation in time $\tilde{O}(mn^{1/2})$ is simply a combination of the \textsc{CycleEstimationDijkstra} data structure with the above standard sampling and truncated Dijkstra techniques.

\paragraph{Time/accuracy trade-off for better running time.} We first focus on algorithms that achieve better than $\tilde{O}(mn^{1/2})$ running time and worse than 2-approximation. Such a result for $n$-PSP follows from the following observation: a 2-approximation algorithm in $\tilde{O}(mn^{1/2})$ time can essentially be plugged into the base case of Thorup-Zwick distance oracles \cite{thorup2005approximate}. This result for $n$-PSP appears in the appendix, however for ANSC, we obtain better results and we focus on those here.

Specifically, for ANSC, we obtain a $(k+\eps)$-approximation in $\tilde{O}(mn^{1/k})$ time (\cref{thm:kapprox-ansc-undir}). A very recent result gave an algorithm for Girth with a similar guarantee. However, our techniques are completely different from theirs. Instead of taking inspiration from an undirected Girth algorithm, we take inspiration from a \emph{directed} Girth algorithm, even though our graph is undirected and techniques for Girth have traditionally been very different in the directed and undirected settings. In general it is not trivial to extend an algorithm for finding cycles from the directed case to the undirected case, as finding undirected cycles introduces challenges that don't appear in directed graphs. To illustrate this, in undirected graphs if we are not careful our algorithm might estimate traversing a path from a node $v$ to $u$ and back to $v$ as a cycle, whereas in the directed case this does not happen. As it turns out, our \textsc{CycleEstimationDijkstra} data structure is useful in addressing this issue.

To obtain our $(k+\eps)$-approximation algorithm for undirected ANSC, we use  \textsc{CycleEstimationDijkstra} together with a labeling procedure similar to our directed ANSC algorithm, which is in turn from the Girth approximation algorithm of \cite{girth-icalp}. These techniques allow us to moderate the size of the vertex sets visited while performing Dijkstra's algorithm and prevent over-processing nodes. Specifically, whenever we do \textsc{CycleEstimationDijkstra}, we only visit nodes with a particular label, and we change the label of a node $v$ when we know that we must have a good enough estimate for $SC(v)$. 

As previously described, finding cycles in undirected graphs presents challenges that are not present for directed graphs, however the opposite is also true; neither setting is clearly strictly harder than the other. We take advantage of the undirected setting to simplify some aspects of the algorithm, which actually yields better bounds for the undirected setting than the directed setting. In particular, for directed graphs, the labeling procedure is used on top of an induction, however we determine that this induction is unnecessary for undirected graphs, and removing it yields an algorithm with better running time and approximation factor.

\paragraph{Time/accuracy trade-off for better approximation ratio.}
On the other side of the time/accuracy trade-off, we consider getting a better that 2-approximation with running time slower than $\tilde{O}(mn^{1/2})$. Such a result was previously known for $n$-PSP \cite{DBLP:conf/esa/Agarwal14}. To get such a result for ANSC, we take inspiration from the algorithm of Dahlgaard, B{\ae}k Tejs Knudsen,  and St{\"o}ckel \cite{dahlgaard2017new} for approximating the girth of a graph. An algorithm very similar to \cite{dahlgaard2017new} carries over from Girth to ANSC. The only qualitative differences are our use of edge sampling instead of vertex sampling and our use of the \textsc{CycleEstimationDijkstra} data structure (\cref{thm:unpubansc}).

\subsubsection{Algorithms in $\tilde{O}(m+n^{2-\eps})$ time}

The approximation algorithms we have mentioned so far have time complexities $m\cdot n^{c}$ for some $c>0$, which are not desirable for very dense graphs. In our next results, our goal is to minimize the dependency of the running times on $m$, while keeping them subquadratic in terms of $n$, that is, we want time complexity $\tilde{O}(m+n^{2-\eps})$ for constant $\eps >0$.

One simple idea is to use spanners to reduce the number of edges to $m' = n^{2-\Omega(1)}$ (while preserving the distances up to some factor), and then apply our previous algorithms to the sparsified graph in $m'\cdot n^c \ll n^2$ time. This idea has been used e.g. in the distance oracle of Wulff-Nilsen \cite{apsp1}. There are many spanner constructions that take only $\tilde{O}(m)$ time, allowing us to obtain $\tilde{O}(m+n^{2-\eps})$ overall time complexity. This simple idea works for the $n$-PSP problem (e.g., \cref{thm:pd-ksquared-2ktime}), but yields quite a large approximation factor, which is the product of the approximation factors of the spanner and the approximation algorithm. For the ANSC problem, this simple idea does not immediately work, since spanners do not give any guarantees on cycle lengths.

\paragraph*{$n$-PSP.}

We briefly explain the idea behind the $(2+\eps,\beta)$-approximation algorithm for $n$-PSP in $\tilde O(m+n^{3/2})$ time for any constant $\eps>0$ and constant $\beta$ depending on $\eps$ (\cref{thm:pd-2approx-16time}). We note that this approximation guarantee is nearly tight with our conditional lower bound in \cref{cor:clique1}).
 
 We take an $O(n^{3/2})$-edge subgraph containing the incident edges of all vertices with degree at most $n^{1/2}$, and compute a $(1+\eps,\beta)$ spanner of size $n^{1+\eps}$ \cite{thorup2006spanners} for this subgraph. Then, we perform the $\tilde O(m\sqrt{n})$-time $2$-approximation algorithm on this spanner. 
 
 Now, it remains to take care of the pairs $(s_i,t_i)$ whose shortest paths pass through some high-degree ($\geq n^{1/2}$) vertex. To do this, we take a sample $S$ of $\sqrt{n}$ nodes, and compute single-source shortest paths from all $s\in S$ in a $(2+\eps,\beta)$-spanner of the graph, and use $\min_{x\in S}\{d(s_i,x)+d(x,t_i)\}$ as a $(2+\eps,\beta')$ estimate of $d(s_i,t_i)$.

\paragraph*{Highlight: ANSC.} As mentioned above, spanners do not provide direct guarantees on cycle lengths. However, we observe that they do give some indirect guarantees as follows. Consider a shortest cycle through $v$, denoted $C_v$. If we divide  $C_v$ into at least three almost equal subpaths where two of the subpaths have $v$ as an endpoint, each subpath is a shortest path. Take one subpath $P$. In any $k$-spanner, there is a subpath $P'$  of length at most $k|P|$ between its endpoints. Since none of these subpaths pass through $v$ except for the ones that start or finish at $v$, the concatenation of all these approximated subpaths contains a cycle that includes $v$ of length roughly $k$ times $SC(v)$.

Even with the above idea, there are still two issues to overcome for ANSC that are not present for $n$-PSP: 
First, when the cycle is very small (length at most $4$) we can't apply the above idea. Second, if our spanner has $t$-additive error and we divide $C_v$ into $s$ subpaths, then the additive error in the estimated cycle length is $t\cdot s$ rather than $t$. 

Our solution to avoid the above issues is to use a \emph{fault-tolerant spanner}.
First, we use a $1$-fault-tolerant $k$-spanner, which is a subgraph $H$ such that for each $u,v$ and edge $e$, if $d$ is the distance between $u$ and $v$ in $G\setminus e$, then $d\le d_H(u,v)\le kd$. 

The following observation illustrates the usefulness of fault-tolerant spanners for this application.
If $P$ is a $1$-fault tolerant $k$-spanner, then for each node $v$, the subgraph consisting of $P$ and all the edges adjacent to $v$ contains a cycle around $v$ of length at most $k\cdot SC(v)$. 

The way we utilize the above observation is as follows. We obtain a sample set $S$ of the nodes, and we perform \textsc{CycleEstimationDijkstra}$(s)$ from each $s\in S$ in a subgraph $G_s$ of the underlying graph $G$. The spanner $P$ is contained in all of the subgraphs $G_s$, and these subgraphs are selected in a way that for each node $v$, all the edges adjacent to $v$ appear in at least one subgraph $G_s$. Moreover, any edge that is not in the spanner $P$ appears in at most $2$ of these subgraphs. This means that from the observation above we get an estimate for $SC(v)$ for each $v$ from one of the \textsc{CycleEstimationDijkstra}s, and our running time is $O(|S|\cdot |E(P)|+m)$. By using a $1$-fault-tolerant $5$-spanner, this idea gives us a $(6,1)$-approximation algorithm for ANSC in time $\tilde{O}(m+n^{2-{1/6}})$. 

Now, our goal is to obtain a better multiplicative approximation factor. 
To develop our $(2+\epsilon,\beta)$-approximation algorithm, we first use $1$-fault tolerant $k$-spanners for large $k$ to get approximation algorithms with running time close to linear. We use this algorithm for estimating small (constant sized) cycles. For bigger cycles, instead of fault tolerant spanners, we use the composition of the  spanners by \cite{k-k-1spanner} and \cite{thorup2006spanners}, together with the observations mentioned at the beginning of this section.

\section{Preliminaries}

\paragraph*{Notations and assumptions.}
Throughout this paper, for weighted graphs, we assume all edge weights are non-negative.

For any vertex $v$, let $N_G(v)$ denote the neighborhood of $v$, and for any positive integer $x$, let $N_G(v,x)$ denote the closest $x$ vertices to $v$, where ties are broken arbitrarily. For any positive integer $r$, let $B_G(v,r)$ be the ball around $v$ of radius $r$; that is, the set of vertices of distance at most $r$ from $v$.
 For all of these notations, if the graph is clear from context, we omit the subscript. 

For a node $v$, Let $C_v$ be the shortest cycle passing through $v$, and let $SC(v)=|C_v|$ be the length of this cycle.

\paragraph*{Known results on spanners.}
Let $G$ be a graph, and let $r$ and $k$ be two positive integers. A subgraph $P$ of $G$ is a $k$-spanner of $G$ if for every two nodes $u,v$, $d_P(u,v)\le kd_G(u,v)$. Baswana and Sen \cite{2k-1spanner} show an $O(km)$-time randomized algorithm to build a $2k-1$ spanner for weighted graphs with at most $O(kn^{1+1/k})$ edges, which was later derandomized by Roditty, Thorup, and Zwick \cite{DBLP:conf/icalp/RodittyTZ05}.

A subgraph $P$ of $G$ is an $r$-fault-tolerant $k$-spanner if for every subset $F$ of nodes with $|F|\le r$, $P$ is a $k$-spanner in the graph $G\setminus F$: for every two nodes $u,v\in G\setminus F$, $d_{P\setminus F}(u,v)\le kd_{G\setminus F}(u,v)$. We only use $1$-fault-tolerant spanners in this paper. Dinitz and Krauthgamer \cite{faulttolerant-spanner} show that any $k$-spanner with at most $f(n)$ edges can be transformed into a $1$-fault-tolerant $k$-spanner with at most $O(\log{n})\cdot f(2n)$ edges. Combining this with the $(2k-1)$-spanner of \cite{2k-1spanner}, we get the following lemma.

\begin{lemma}\label{lem:fault-tolerant-const}
Given an $n$-node $m$-edge graph $G$ and an integer $k>1$, one can build a $1$-fault-tolerant $2k-1$ spanner of size at most $O(kn^{1+1/k})$ in $\tilde{O}(m)$ time. 
\end{lemma}

In an unweighted graph $G$, an $(\alpha,\beta)$-spanner is a subgraph $P$ where for every two nodes $u,v\in V(G)$, we have $d_P(u,v)\le \alpha d_G(u,v)+\beta$. Baswana, Kavitha, Mehlhorn, and Pettie \cite{k-k-1spanner} construct a $(k,k-1)$-spanner of size $O(n^{1+1/k})$ for an $n$-node $m$-edge graph in $O(m)$ time.

We also use the following spanner of Thorup and Zwick \cite{thorup2006spanners}.
\begin{theorem}[\cite{thorup2006spanners}]
\label{thm:very-sparse-spanner}
Let $G$ be an undirected unweighted graph with $n$ nodes and $m$ edges. For every integer $t\ge 2$, a spanner $P$ with $O(tn^{1+1/t})$ edges can be constructed in $O(mn^{1/t})$ time where such that for every $u,v\in V(G)$, $d_{G}(u,v)\le d_P(u,v)\le d_{G}(u,v)+O(d_G(u,v)^{1-\frac{1}{t-1}})$.
\end{theorem}

\begin{corollary}
\label{cor:very-sparse-spanner}
For every constant $\epsilon>0$, there is a constant $\beta>0$ such that the following holds:
given an undirected unweighted graph $G$ with $n$ nodes and $m$ edges, there exists a $(1+\eps,\beta)$-spanner for $G$ with $n^{1+\epsilon}$ edges that can be computed  $O(mn^{\eps})$ time.
\end{corollary}
\begin{proof}
 We apply Theorem \ref{thm:very-sparse-spanner} to the input graph $G$ with $t = \lceil   1/\eps \rceil$  to construct a spanner $H_{\eps}$ in $O(m n^{1/t}) \le O(m n^\eps)$ time, which has at most $O(tn^{1+1/t}) \le O(n^{1+\eps})$ edges.
 Let $\beta$ be some constant depending on $\eps$ to be determined later.

  Consider two nodes $u,v$ and let $d_{G}(u,v)=d\ge 1$. Then $d_{H_\epsilon}(u,v)\le d+O(d^{1+1/(t-1)})<d+c\cdot d^{1-1/t}$ for some constant $c$ independent of $u$ and $v$. Let $D = (c/\eps)^t$. For $d\ge D$, we have $d_{H_\eps}(u,v) \le d+c\cdot d^{1-1/t} \le d(1+c\cdot D^{-1/t}) = (1+\eps)d$. 
  For $d\le D$, we have $d_{H_\eps}(u,v) \le d+c\cdot d^{1-1/t} \le D+c\cdot D^{1-1/t}$.
 Hence,   $d\le d_{H_\eps}(u,v) \le (1+\eps)d+\beta$ holds for $\beta = D+c\cdot D^{1-1/t}$ for all $d\ge 1$. So $H_{\eps}$ is a $(1+\epsilon,\beta)$ spanner of $G$.
\end{proof}

Combining the two spanner constructions of Baswana et al.\ \cite{k-k-1spanner} and Thorup and Zwick \cite{thorup2006spanners}, we have the following lemma. 
\begin{lemma}
\label{lem:lineartime-sparse}
 Let $G$ be an undirected unweighted graph with $n$ nodes and $m$ edges. For any arbitrarily small constant $\eps>0$, one can construct a $(2+\eps,\beta)$-spanner $H_\eps$ of $G$ with $O(n^{1+\eps/2})$ edges in $O(m + n^{1.5+\eps/2})$ time, for some constant $\beta$ depending on $\eps$.
\end{lemma}
\begin{proof}
 Let $H_{2,1}$ be the $(2,1)$-spanner with $m' = O(n^{1.5})$ edges on graph $G$,  which can be constructed in $O(m)$ time \cite{k-k-1spanner}.
Next, we apply \cref{cor:very-sparse-spanner} on $H_{2,1}$, and obtain a $(1+\eps/2,\beta')$-spanner $H_{\eps}$ of $H_{2,1}$ with $O(n^{1+\eps/2})$ edges in $O(m'n^{\eps/2})\le O(n^{1.5+\eps/2})$ time.
Composing the approximation guarantees, we can see that $H_{\eps}$ is a $(2+\eps,2\beta'+1)$-spanner of $G$.
\end{proof}

\section{$n$-Pairs Shortest Paths}
\input{pairwise_distances}

\section{All-Nodes Shortest Cycles}
\input{ANSC}

\bibliographystyle{alpha}
\bibliography{bib} 
\input{appendix}
\end{document}

%% file: pairwise_distances.tex
\subsection{Conditional Lower Bounds}



First we provide another reduction from Dense Triangle Hypothesis to $n$-PSP. 

\twoapproxlbpd*
\begin{proof}
Let $G=(V,E)$ be an instance of Triangle Detection. We construct a $3$-layered graph $G'$ with layers $V_1,V_2,V_3$, where $V_i$ is a copy of $V$, and $v_i\in V_i$ is a copy of $v\in V$ for all $i=1,2,3$. For all $(u,v)\in E$, we add an edge between $u_1$ and $v_2$, and an edge between $u_2$ and $v_3$. Finally add $n^2$ single nodes to $G'$. So $G'$ has $N=n^2$ nodes and edges. 

Consider the $N$ pairs $(v_1,u_3)$ for all $v,u\in V$. If $G$ contains a triangle $vwu$, then $d_{G'}(v_1,u_3)=2 $. If there is no triangle passing through $v$ and $u$, then $d_{G'}(v_1,u_3)\ge 4$ since $G'$ is bipartite. So if $G$ has no triangle, then for all $u,v$, $d_{G'}(v_1,u_3)\ge 4$.

So if for $\epsilon>0$ there is a $(2-\epsilon)$-approximation algorithm for $n$-Pairs Minimum Distance in $O(n^{3-\delta})=O(N^{3/2-\delta/2})$ time for some $\delta>0$, then this algorithm can distinguish between the case where $G$ has a triangle and the case where $G$ has no triangle.
\end{proof}

Next we use the $k$-Clique and $(k,r)$-Hyperclique Hypotheses to show lower bounds  for the $n$-Pairs Minimum Distance problem.

\cliquereduction*

\begin{proof}

We will present the proof for $r=2$. Generalization to $r$-uniform hypergraphs with $r\ge 3$ is straightforward and will be briefly described at the end of the proof.

First note that we can assume the $k$-clique instance is on a $k$-partite graph with vertex partition $V=V_1\cup V_2\cup \dots\cup V_k$ where $|V_i|=n$. This is because we can create $k$ copies $V_1,\ldots,V_k$ of the vertex set of $G$, and if $uv$ is an edge in $G$, we put an edge between the copy of $u$ in $V_i$ and the copy of $v$ in $V_{j}$ for all $i,j\in \{1,\ldots,k\}$.

We define an $n$-Pairs Minimum Distance instance on a new undirected graph $G'$ as follows: $G'$ is $(k+1)$-partite with node partitions (for clarity we refer to vertices in graph $G'$ as ``nodes'') $U_1,U_2,\dots,U_{k+1}$ where $|U_i|=n^t$.  
For every $1\le i\le k+1$, we assume a natural bijection between $U_i$ and $V_i \times V_{i+1}\times \dots \times  V_{i+t-1}$ (where indices in the subscripts of $V$ are modulo $k$), and hence will denote the nodes in $U_i$ as $t$-tuples.

For every $1\le i\le k$, we add edges between $U_i$ and $U_{i+1}$ as follows: for every $v_i\in V_i,v_{i+1}\in V_{i+1},\dots,v_{i+t}\in V_{i+t}$ that form a $(t+1)$-clique in $G$, we add an edge between node $(v_i,v_{i+1},\dots,v_{i+t-1}) \in U_i$ and node $(v_{i+1},v_{i+2},\dots,v_{i+t}) \in U_{i+1}$ in graph $G'$.
Overall, $G'$ has $(k+1)\cdot n^{t}$ nodes and at most $k\cdot n^{t+1}$ edges.

Finally, for every $v_1\in V_1,\dots,v_{t}\in V_{t}$, we add a query node pair consisting of $(v_1,\dots,v_{t})\in U_1$ and $(v_1,\dots,v_{t})\in U_{k+1}$ into this $n$-Pairs Minimum Distance instance. The number of input pairs is $n^{t}$.

Now we will prove this reduction satisfies the claimed properties \ref{claim1} and \ref{claim2} in the statement.

Property~\ref{claim1} is easy to prove: if $G$ contains a $k$-clique $(v_1,v_2,\dots,v_{k})$, then by definition there is an edge between $(v_i,\dots,v_{i+t-1})\in U_i$ and $(v_{i+1},\dots,v_{i+t})\in U_{i+1}$ for every $1\le i\le k$, and hence the shortest path from $(v_1,\dots,v_{t})\in U_1$ to $(v_1,\dots,v_{t})\in U_{k+1}$ in $G'$ has length $k$, going through $U_1\to U_2\to \dots \to U_{k} \to U_{k+1}$.

It remains to prove Property~\ref{claim2}. Take any path $p$ going from
$(v_1,\dots,v_{t})\in U_1$ to $(v_1,\dots,v_{t})\in U_{k+1}$ in $G'$.
For convenience of the following arguments, from now on we identify $U_1$ with $U_{k+1}$, so that the two endpoints of path $p$ are glued together, turning $p$ into a closed walk.  We think of $U_i\to U_{i+1}$ (from now on the indices under $U$ are also considered mod $k$) as the ``positive direction'' of the cycle. If $p$ had length exactly $k$, then $p$ would have to go through (along the positive direction) the cycle $U_1\to U_2\to \dots \to  U_k \to U_1$, which would immediately imply the existence of a $k$-clique in the original graph $G$. In reality, $p$ can be longer than $k$, allowing it to go backwards ($U_{i}\to U_{i-1}$) in some of the steps, but effectively it still travels around the cycle $U_1\to U_2\to \dots \to  U_k \to U_1$ once in the positive direction (in other words, it has ``winding number'' $+1$). Our goal is to show that such a closed walk $p$ should still imply the existence of a $k$-clique in $G$, as long as $p$ has length strictly smaller than $4(t+1)-k$.

For every $1\le i\le k$, the cyclic structure of $p$ implies that $p$ must contain a contiguous subpath $q_i$ that satisfies the following: 
\begin{itemize}
\item Path $q_i$ starts from some node in $U_{i-t}$ and ends at some node in $U_{i+1}$, and,
\item All internal nodes of $q_i$ (i.e., those that are not the starting node or the ending node) belong to $U_{i-t+1}\cup U_{i-t+2}\cup \dots\cup U_{i}$.
\end{itemize}
We take such a path $q_i$, and focus on the internal nodes of $q_i$ (note that there is at least one internal node). Recall that $U_{i-t+1}$ is in bijection with $V_{i-t+1}\times \dots \times V_{i}$, $U_{i-t+2}$ is in bijection with $V_{i-t+2}\times \dots \times  V_{i+1}$, $\dots$, and $U_{i}$ is in bijection with $V_{i}\times \dots \times V_{i+t-1}$.
Hence, every internal node of path $q_i$ is a $t$-tuple that has one coordinate being a $V_i$-vertex. Moreover, from the definition of the edges between $U_{i-t+1},U_{i-t+2},\dots,U_i$, any two adjacent internal nodes should share the same $V_i$-vertex in their $t$-tuple representations, so we can conclude that all these $V_i$-vertices actually must be the same vertex, denoted as $v_i^* \in V_i$. As a consequence, every edge on the path $q_i$ corresponds to a $(t+1)$-clique in $G$ that includes vertex $v_i^*$.

We will show that $v_1^*\in V_1,\dots, v_k^* \in V_k$ form a $k$-clique in $G$. Suppose they do not, then there are two $v_i^*,v_j^*$ that are not adjacent in $G$. Now we inspect the two subpaths $q_i,q_j$ of the closed walk $p$. Observe that $q_i,q_j$ do not share an edge, since otherwise this edge would correspond to a $(t+1)$-clique in $G$ that include both $v_i^*$ and $v_j^*$.

Now, for $1\le l\le k$, let $c_l$ denote the number of times $p$ goes through an edge between $U_l$ and $U_{l+1}$ (in either direction). Then the length of $p$ equals $\sum_{l=1}^k c_l$, and every $c_l$ has to be a positive odd integer.  Since $q_i$ contains at least one edge of the form $U_l\to U_{l+1}$ for every $l\in \{i-t,i-t+1,\dots,i\}$, and $q_j$ contains at least one edge of the form $U_l\to U_{l+1}$ for every $l\in \{j-t,j-t+1,\dots,j\}$, we know that $c_l>1$ for every $l\in \{i-t,i-t+1,\dots,i\} \cap \{j-t,j-t+1,\dots,j\}$ due to the disjointness of $q_i$ and $q_j$. In other words, $c_l\ge 3$ holds for at least $(t+1)+(t+1)-k$ many $l$'s. Hence, the length of $p$ is at least 
\[ (3-1) \cdot \big ((t+1)+(t+1)-k\big ) + k = 4(t+1)-k,\]
contradicting our assumed upper bound on its length.

Now we discuss the generalization to the case where $G$ is an $r$-uniform hypergraph ($r\ge 3$). Our construction of the graph $G'$ remains the same: for every $v_i\in V_i,v_{i+1}\in V_{i+1},\dots,v_{i+t}\in V_{i+t}$ that form a $(t+1)$-hyper-clique in hypergraph $G$, we add an edge between node $(v_i,v_{i+1},\dots,v_{i+t-1}) \in U_i$ and node $(v_{i+1},v_{i+2},\dots,v_{i+t}) \in U_{i+1}$ in graph $G'$. The only change is in the last paragraph of the argument: assume $v_1^*,\dots, v_k^*$ do not form a $k$-hyperclique, then there exist $v_{i_1}^*,\dots,v_{i_r}^*$ that do not form a hyperedge. Again, this implies that the paths $q_{i_1},\dots,q_{i_r}$ cannot have common intersection. Since 
\[ \big \lvert\bigcap_{j=1}^r \{i_j-t,i_j-t+1,\dots,i_j\}\big \rvert \ge k - r\big (k-(t+1)\big ),\]
we similarly conclude that the length of $p$ is at least 
\[ (3-1)\cdot \Big (k - r\big (k-(t+1)\big )\Big ) + k= 2r(t+1) - (2r-3)k .\]
\end{proof}

\cref{lem:cliquereduction} has the following implications.

\clique*
In particular,  \cref{cor:clique1} says better-than-$2$-approximation cannot be done in significantly faster than $m\sqrt{n}$ time or $n^2$ time.
\begin{proof}[Proof of \cref{cor:clique1}]
Applying \cref{lem:cliquereduction} with $r=2$ and  $t=k-2$, we can see that solving an instance with $n = \Theta(n_0^t)$ vertices and $m = \Theta(n_0^{t+1}) = \Theta(n^{(t+1)/t})$ edges with approximation ratio better than $\frac{4(t+1)-k }{ k} = 3-4/k$ can be used to solve $k$-clique on $n_0$-vertex graphs, which requires $n_0^{k-o(1)} = n^{k/t - o(1)}$ time combinatorially. So  there cannot exist a combinatorial algorithm in $n^{k/(k-2)-\eps} = n^{k/t-\eps}$ time, or in $m\cdot n^{1/(k-2)-\eps } = n^{(t+1)/t + 1/(k-2)-\eps }= n^{k/t-\eps}$ time.
\end{proof}

\cliquee*
\begin{proof}
Letting $k=4k'$ and applying \cref{lem:cliquereduction} with $r=2$ and $t=\lfloor k/2 \rfloor = k/2$, we can see that solving an instance with $n = \Theta(n_0^t)$ vertices and $m = \Theta(n_0^{t+1}) = \Theta(n^{(t+1)/t})$ edges with approximation ratio better than $\frac{4(t+1)-k }{ k} = 1+1/k'$  can be used to solve $k$-clique on $n_0$-vertex graphs, which requires $n_0^{k-o(1)} = n^{2 - o(1)}$ time combinatorially. So  there cannot exist a combinatorial algorithm in $m^{2-2/(k'+1)}\cdot n^{1/(k'+1)-\eps} = n^{2-\eps} $ time, or in $m\cdot n^{1-1/(2k')-\eps } = n^{(k+2)/k + (k-2)/k-\eps} = n^{2-\eps}$ time.
\end{proof}

\cliqueee*
\begin{proof}
Letting $k=4k'+2$ and applying \cref{lem:cliquereduction} with $r=2$ and $t=\lfloor k/2 \rfloor = k/2$, we can see that solving an instance with $n = \Theta(n_0^t)$ vertices and $m = \Theta(n_0^{t+1}) = \Theta(n^{(t+1)/t})$ edges with approximation ratio better than $\frac{4(t+1)-k }{ k} = 1+1/(k'+0.5)$  can be used to solve $k$-clique on $n_0$-vertex graphs, which requires $n_0^{k-o(1)} = n^{2 - o(1)}$ time combinatorially. So  there cannot exist a combinatorial algorithm in $m^{2-3/(k'+2)}\cdot n^{2/(k'+2)-\eps} = n^{2-\eps} $ time, or in $m\cdot n^{1-1/(2k'+1)-\eps } =  n^{2-\eps}$ time.
\end{proof}

Notably, the conditional lower bounds of \cref{cor:clique2} and \cref{cor:clique3} almost match the algorithms in  \cite{DBLP:conf/esa/Agarwal14}.

Similarly, we can also use the $(k,r)$-Hyperclique Hypothesis (for $r\ge 3$) to get non-combinatorial lower bounds.


   \hyperr*
\begin{proof}
Directly follows from \cref{lem:cliquereduction} by setting $r=3,t=k-2$.
\end{proof}


\subsection{Approximation Algorithm for $n$-PSP}

\almosttwopdundir*

\begin{proof}\footnote{This proof has been simplified from the old version, which additionally used the emulator from \cite{cz22}. We thank the anonymous reviewer for pointing out that using the emulator from \cite{cz22} was unnecessary. }
Let $S$ be a set of sampled nodes of size $\tilde O(n^{1/2})$. Let $N(S)$ denote the set of all nodes that are adjacent to some node in $S$. Then, with high probability, all $v\in V - N(S)$ have degree at most $n^{1/2}$. Define the edge set $\mathcal{E}_S:=  \bigcup_{v\in V-N(S)} E(v)$, where $E(v)$ denotes all the incident edges of $v$. Then, we have $|\mathcal{E}_S| \le n^{3/2}$ with high probability.

For every given vertex pair $(u,v)$, our estimate for $d_G(u,v)$ will be $\min \{d_1(u,v),d_2(u,v)\}$, where $d_1(u,v),d_2(u,v)$ are computed as follows.
\paragraph*{Part 1.}
We use \cref{cor:very-sparse-spanner} to compute an $(1+\eps/2, \beta/2)$-spanner of the subgraph with edges from $\mathcal{E}_S$. This spanner contains $n^{1+\eps/2}$ edges and can be computed in $O(|\mathcal{E}_S|\cdot n^{\eps/2}) \le O(n^{1.5+\eps/2})$ time. Then, we run the algorithm from \cref{thm:pd-k-tz} with approximation ratio $k=2$ on this spanner, which implies $(2+\eps,\beta)$-approximation of $d_{\mathcal{E}_S}(u,v)$ for all input pairs in $\tilde O(n^{1+\eps/2}\cdot n^{1/2}) \le \tilde O(n^{1.5+\eps/2})$ time. We denote this estimate as $d_1(u,v)$.
     
     Note that $d_1(u,v)$ gives $(2+\eps,\beta)$-approximation of $d_G(u,v)$ if the shortest path from $u$ to $v$ in $G$ completely lies in $\mathcal{E}_S$. We will handle the remaining cases in Part 2.

\paragraph*{Part 2.} Now, we can assume the shortest path from $u$ to $v$ in $G$ contains an edge $(x',y')$ with both endpoints $x',y' \in N(S)$.

Use  \cref{lem:lineartime-sparse} to build an $(2+\eps,\beta)$-spanner $H$ of the graph $G$ with $O(n^{1+\eps/2})$ edges in  $\tilde O(m+n^{1.5+\eps/2})$ time. For every $s\in S$, we use
$\tilde O(|H|) \le \tilde O(n^{1+\eps/2})$ time to compute the single source distances $d_H(s,v)$ for all $v\in V$. The total time is $|S|\cdot \tilde O(n^{1+\eps/2}) \le \tilde O(n^{1.5+\eps/2})$. Then, for every given vertex pair $(u,v)$, our estimate for $d_G(u,v)$ is 
\[ d_2(u,v) = \min_{s\in S}\{d_{H}(u,s) + d_{H}(s,v)\}.\]
Let $s'\in S$ be adjacent to $x'$ in $G$. We have
\begin{align*}
    d_2(u,v) & \le d_{H}(u,s') + d_{H}(s',v)\\
    & \le (2+\eps)d_G(u,s') + \beta + (2+\eps)d_G(s',v) + \beta\\
    & \le (2+\eps)(d_G(u,x')+1 + d_G(x',v)+1) + 2\beta\\
    & = (2+\eps)d_G(u,v) + (4+2\eps + 2\beta).
\end{align*}
Hence, we have shown that $d_2(u,v)$ provides a good estimate.
\end{proof}


%% file: ANSC.tex


\subsection{\textsc{CycleEstimationDijkstra} Data Structure}

We use the following lemma, which directly follows from the Link-Cut Tree data structure by Sleator and Tarjan \cite{DBLP:journals/jcss/SleatorT83}.
\begin{lemma}
There is a data structure that maintains an $n$-vertex tree with vertex weights $\{c_v\}_{v\in V}$, supporting the following updates and queries in $O(\log n)$ time per operation:
\begin{itemize}
    \item $\textsc{Update}(v_1,v_2,x)$: Given vertices $v_1,v_2$ and weight $x$, perform $c_u \gets \min \{c_u, x\}$ for all vertices $u$ on the path from $v_1$ to $v_2$ in the tree.
    \item $\textsc{Query}(v)$: Given a  vertex $v$, return its weight $c_v$.
\end{itemize}
This data structure can be initialized in $O(n\log n)$ time.
\label{lem:lct}
\end{lemma}
\begin{proof}[Proof Sketch]
We use the Link-Cut Tree data structure by Sleator and Tarjan \cite{DBLP:journals/jcss/SleatorT83} (with modifications (4) and (5) mentioned in \cite[Page 365]{DBLP:journals/jcss/SleatorT83}). Note that our $\textsc{Update}(v_1,v_2,x)$ operation can be implemented as two operations, $evert(v_1)$ and then $update(v_2,x)$, in the interface of the original data structure \cite{DBLP:journals/jcss/SleatorT83}.
\end{proof}


The following algorithm \textsc{CycleEstimationDijkstra} combines (truncated) Dijkstra's algorithm and \cref{lem:lct}, which approximates  the shortest cycle length $SC(v)$ for every vertex $v$. In different applications of this algorithm we will truncate Dijkstra's algorithm in different ways e.g. after exploring a certain number of edges or vertices, or after reaching a certain radius.

\begin{definition}[\textsc{CycleEstimationDijkstra} algorithm]
Given an undirected graph $G$ with $n$ vertices and $m$ weighted edges, and a starting vertex $s$ in $G$, the algorithm \textsc{CycleEstimationDijkstra}$(s)$ performs the following steps:
\begin{enumerate}
    \item Run (truncated) Dijkstra's algorithm on $G$ starting from $s$, to compute a (partial) shortest-path-tree $T$ rooted at $s$.
    \item Initialize the data structure from \cref{lem:lct} on tree $T$, with initial weights $c_v=+\infty$ for all vertices $v$.
    \item  For every explored edge $e=(u,v)$ that is not in $T$, perform $\textsc{Update}(u,v, w_e+d_{T}(u,v))$.
    \label{itemstep3}
    \item Return $c_v$ (obtained from calling $\textsc{Query}(v)$) for all vertices $v$ that were visited by the (truncated) Dijkstra's algorithm).
\end{enumerate}
The running time of \textsc{CycleEstimationDijkstra}$(s)$ is $\tilde O(M)$ where $M$ is the number of edges explored in the (truncated) Dijkstra's algorithm.
\end{definition}

The following lemma states the approximation guarantee of the  \textsc{CycleEstimationDijkstra} algorithm. 

\begin{lemma}
\label{lem:datastructure}
In an undirected graph, let $x,s$ be two nodes, and let $C$ be a cycle passing through $x$. Suppose \textsc{CycleEstimationDijkstra}$(s)$ explored all edges of $C$.  Then, for \textbf{all} nodes $y$ on the cycle $C$, the return values satisfy $SC(y)\le c_{y} \le 2d(s,x)+|C|$, where $|C|$ denotes the total weight of $C$.
\end{lemma} 

\begin{proof}
First, observe that $c_{y}$ must equal the length of some cycle in the graph passing through $y$,  since in Step~\ref{itemstep3} of \textsc{CycleEstimationDijkstra}, all the vertices whose weights get updated lie on the cycle consisting of edge $(v,u)$ and the path $u\to v$ on tree $T$, with total weight $w_e+d_{T}(u,v)$. Hence, $c_{y}\ge SC(y)$.

It remains to show $c_{y} \le 2d(x,s)+|C|$ for all nodes $y\in C$, which immediately follows from the following two claims:
\begin{itemize}
    \item \textbf{Claim 1:} For every node $y\in C$, there exists an edge $(u,v)$ on cycle $C$ that is not in $T$, such that $y$ lies on the path from $u$ to $v$ on $T$.
    \item \textbf{Claim 2:} For every edge $e=(u,v)$ on cycle $C$ that is not in $T$, we have $w_e+d_T(u,v)\le 2d(s,x)+|C|$.
\end{itemize} 
\begin{proof}[Proof of Claim 1]
Replace every non-tree edge $(u,v)$ on $C$ by the path from $u$ to $v$ on the tree $T$, and obtain a closed walk $C'$ whose edges all belong to $T$.  Since $y\in C$, some incident edge $(y,x)$ must appear in $C'$. Observe that every edge in $T$ is traversed by $C'$ an even number of times, and,  in particular, $(y,x)$ is traversed by $C'$ at least twice. Since $C$ only passes through $(y,x)$ at most once, there must be another non-tree edge $(u,v)$ on $C$ such that $(y,x)$ lies on the tree path from $u$ to $v$.
\end{proof}
\begin{proof}[Proof of Claim 2]
Without loss of generality, we assume node $x$ is the minimizer of $\min_{x\in C}d(s,x)$.
Divide the cycle $C$ into three parts, $v\to x$, $x\to u$, and edge $e=(u,v)$. Let $w_{vx},w_{xu}$ denote the length of the first two parts. We have 
\begin{align*}
    |C| -w_e&= w_{vx}+w_{xu}\\
    & \ge d(v,x)+d(u,x)\\
    & \ge \big (d(s,v)-d(s,x)\big ) + \big (d(s,u)-d(s,x)\big )\\
    & \ge d_T(u,v) - 2d(s,x).
\end{align*}
\end{proof}

\end{proof}
\subsection{Approximation Algorithms for ANSC}
\subsubsection{Warm-up: 2-approximation for ANSC in $\tilde{O}(m\sqrt{n})$ time}

\begin{theorem}
\label{thm:warmupp}
Given an $n$-node $m$-edge undirected weighted graph $G$, there is a randomized algorithm for ANSC that computes a 2-approximation in $\tilde{O}(m\sqrt{n})$ time with high probability.
\end{theorem}

\begin{proof} The algorithm and analysis are as follows.
\paragraph{Algorithm.} Let $S$ be a random sample of $\Theta(\sqrt{n}\log(n))$ edges. With high probability, for every vertex $v$, $S$ hits among the closest $m/\sqrt{n}$ edges to $v$, where the distance from $v$ to an edge $(u,u')$ is defined as $\min\{d(v,u),d(v,u')\}$ (and ties are broken arbitrarily).

We run \textsc{CycleEstimationDijkstra}$(s)$ for all endpoints $s$ of edges in $S$. Then, for every vertex $v$, we run a truncated version of \textsc{CycleEstimationDijkstra}$(v)$, where we stop after seeing $m/\sqrt{n}$ edges. 

For all vertices $v$, our estimate $\hat{c}_v$ for the shortest cycle through $v$ is the minimum estimate obtained over all executions of \textsc{CycleEstimationDijkstra}.

\paragraph{Analysis of correctness.} We will show that for all vertices $w$, we have $\hat{c}_w\leq 2\cdot SC(w)$. Fix $w$ and a shortest cycle $C_w$ through $w$.

\subparagraph{Case 1: $\min_{(u,v)\in S}\{d(w,u),d(w,v)\}\leq SC(v)/2$.} In this case, by \cref{lem:datastructure}, we have $\hat{c}_w\leq 2\cdot SC(w)$.

\subparagraph{Case 2: Otherwise.} Because $S$ hits among the closest $m/\sqrt{n}$ edges to $w$, and $S$ does not hit any vertex with an endpoint within distance $SC(v)/2$ of $w$, we know that the closest $m/\sqrt{n}$ edges to $w$ contain all edges with at least one endpoint at distance at most $SC(v)/2$ from either $w$. That is, truncated \textsc{CycleEstimationDijkstra}$(w)$ explores all edges with at least one endpoint within distance $SC(v)/2$ from $w$, which includes all edges on $C_w$. Thus, $\hat{c}_w=SC(w)$.

\paragraph{Analysis of time complexity.} The total time for running \textsc{CycleEstimationDijkstra}$(s)$ for all $s\in S$ is $\tilde{O}(m\sqrt{n})$. The total time for running truncated \textsc{CycleEstimationDijkstra} from each vertex is $O(n\cdot m/\sqrt{n})=O(m\sqrt{n})$.
\end{proof}

\input{ansc2}

%% file: ansc2.tex
\subsubsection{($k+\eps$)-approximation for ANSC in $\tilde{O}(mn^{1/k})$ time}

\kanscundir*
\begin{proof}
We are going to assume that vertex degrees are uniform, i.e. for each node $v$, $d(v)=O(m/n)$. We can achieve the degree uniformity without changing the length of any cycles as follows: For each vertex $v$, replace $v$ and the edges attached to it with a balanced tree with vertices of degree $O(m/n)$ and $d(v)n/m$ leaves. We assign the original edges adjacent to $v$ to the leaves such that each leaf has degree $O(m/n)$. The internal edges of the tree will have weight zero. Note that in this process we are not creating or destroying any cycles, and any cycle in the old graph corresponds to a cycle in the new graph of the same length. For each $v$ we create $O(d(v)n/m)$ new nodes and edges, so we add $O(n)$ nodes and edges to the graph in total.
We call the nodes that are in the original graph, \emph{original nodes}. 



Suppose that for a fixed $D$, we want to know for each $v$ if there is a cycle of length at most $kD$ passing through $v$ or if the shortest cycle passing through $v$ has length larger than $D/(1+\epsilon)$. Then by setting $D=(1+\epsilon)^i$ for $i=1,\ldots,\log_{1+\epsilon}nM$, we can get a $k+O(\epsilon)$-approximation for the length of the shortest cycle of each node. 

We initially label all the nodes ``on". Let $S$ be a sample set of $\tilde{O}(n^{1/k})$ nodes, to hit the closest $n^{\frac{k-1}{k}}$ nodes to each node in the graph. We do  \textsc{CycleEstimationDijkstra}$(s)$ for each $s\in S$. For each node $v$, let $p(v)\in S$ be the closest node to $v$ from $S$.


If $d(v,p(v))\le (k-1)D/2$, we mark $v$ ``off". Note that if there is a cycle of length at most $D$ passing through such $v$, the estimate that the Dijkstra from $p(v)$ gives for $SC(v)$ is at most $D+2\cdot (k-1)D/2 = kD$ by Lemma \ref{lem:datastructure}.

We do the following until there are no on nodes left among the original nodes:

We take an original on node $v$ and do  \textsc{CycleEstimationDijkstra$(v)$}, only visiting on nodes, until we see $n^{1/k}$ nodes. Let $r_1(v)$ be the distance of the farthest node that we visited from $v$. If $r_1(v)\ge D/2>\frac{D}{2(1+\epsilon)}$, we know if the shortest cycle through $v$ is bigger or smaller than $D/(1+\epsilon)$, and we mark $v$ off. This is because if $SC(v)\le D/(1+\epsilon)$, then $C_v$ is contained in the nodes we visited in the Dijkstra\footnote{Note that because we assume that $r_1(v)\ge D/2$ which is strictly bigger than $D/2(1+\eps)$, if $SC(v)\le D/(1+\epsilon)$ we have visited all the nodes in $C_v$ connected to zero edges as well.}, and by Lemma \ref{lem:datastructure} we get an estimate of $SC(v)$ for $v$. In this case we have spent $\frac{m}{n}n^{1/k}$ time and have marked one node off. 

So assume that $r_1(v)<D/2$. Let $r_i(v)$ be the distance of the farthest node to $v$ among the $n^{i/k}$ closest on nodes to $v$. We continue doing  \textsc{CycleEstimationDijkstra}$(v)$ up to radius $r_{i+1}(v)$ such that $r_i(v)<iD/2$ and $r_{i+1}\ge (i+1)D/2$. We show that such $i$ exists. If $|B(v,{(k-1)D/2})|> n^{(k-1)/k}$, $S$ would hit $B(v,{(k-1)D/2})$ with high probability and we would have $d(p(v),v)\le (k-1)D/2$ which is a contradiction to the assumption that $v$ is on. So we have that $|B(v,{(k-1)D/2})|\le n^{(k-1)/k}$ and so $r_{k-1}(v)\ge (k-1)D/2$. Since $r_1(v)<D/2$, such $i$ exists. We mark all the nodes in the ball of radius $iD/2$ off. Recall that there are at least $n^{i/k}$ nodes in this ball. This finishes the description of the algorithm.

Note that for a node $u$, all the  \textsc{CycleEstimationDijkstra}s visiting $u$ give an estimate for $SC(u)$, and we only need to keep the smallest of these estimates. To analyse the running time of the last step, we have spent $\frac{m}{n}n^{(i+1)/k}$, and we have marked at least $n^{i/k}$ nodes off, so we have spent at most $\frac{m}{n}n^{1/k}$ per node, which adds up to $\tilde{O}(mn^{1/k})$ total running time.

We now reason why whenever we label a node off, we know whether its shortest cycle has length bigger than $D/(1+\epsilon)$ or at most than $kD$.
Consider a node $u$ with $SC(u)\le D/(1+\epsilon)<D$, and suppose that the first node that is set off in $C_u$ is $w$. Note that we might have $w=u$. We show that one of the estimates the algorithm gives to $u$ is at most $kD$. We have a few cases:

\begin{enumerate}
    \item If $w$ was set off because $d(w,p(w))\le (k-1)D/2$, then \textsc{CycleEstimationDijkstra}$(p(w))$ gives an estimate of $2\cdot (k-1)D/2 + SC(u)\le kD$ for all the nodes in $C_u$ by Lemma \ref{lem:datastructure}.
    \item If $w$ was set off because we did  \textsc{CycleEstimationDijkstra}$(w)$ and $r_1(w)\ge D/2$, then $C_u$ is contained in the nodes visited, and so we give an estimate of at most $SC(u)\le kD$ to all the nodes in $C_u$ containing $u$ by Lemma \ref{lem:datastructure}.
    \item If $w$ was set off because $w$ is in the ball of radius $iD/2$ of on nodes around some node $v$ for some $i<k-1$, $r_i(v)<iD/2$ and $r_{i+1}(v)\ge (i+1)D/2$, then since $SC(u)\le D/(1+\epsilon)$, all of the nodes of $C_u$ are among the nodes visited up to radius $r_{i+1}(v)$. So the Dijkstra from $v$ gives an estimate of $2\cdot iD/2+SC(u)\le (i+1)D< kD$ to all of the nodes in $C_u$, including $u$.
\end{enumerate} 
\end{proof}

\subsubsection{Approximation algorithms with $\tilde{O}(m+n^{1-\epsilon})$ running time}


We are going to present two algorithms for undirected ANSC in this section. First a rather simple $(6,1)$-approximation algorithm, and then a $(2+\epsilon,\beta)$ approximation algorithm where $\epsilon$ is any small constant and $\beta$ is a function of $\epsilon$. 

We are going to use $1$-fault-tolerant $(2k-1)$-spanners of Lemma \ref{lem:fault-tolerant-const} in this section. The following lemma lays out how fault tolerant spanners can help with cycle estimation.
\begin{lemma}
\label{lem:cycle-est-fault-tolerant}
Let $P$ be a $1$-fault-tolerant $k$-spanner of a graph $G$. For any vertex $v$, let $E_G(v)$ be the set of edges in $G$ attached to $v$, and let $H$ be the graph on the nodes of $G$ with edges in $P$ and $E_G(v)$. Then there is a cycle of length at most $k\cdot SC(v)$ in $H$ that passes through $v$. 
\end{lemma}
\begin{proof}
Let $u$ be a neighbor of $v$ on $C_v$, and let $w\not \in \{ u,v\}$ be a node on $C_v$. Consider the shortest path $P_1$ between $v$ and $w$ in $P\setminus uv$, and the shortest path $P_2$ between $u$ and $w$ in $P\setminus uv$. So paths $P_1,P_2$ do not use the edge $uv$.
Let $P'_1$ be the subpath in $C_v$ from $v$ to $w$ that does not contain $uv$, and let $P'_2$ be the subpath in $C_v$ from $u$ to $w$ that does not contain $uv$. By the definition of 1-fault-tolerant $k$-spanner, we have $|P_1|\le k|P'_1|$ and $|P_2|\le k|P'_1|$. The union of $P_1$, $P_2$ and $uv$ contains a cycle around $v$ in $H$, and its length is at most $|P_1|+|P_2|+1\le k(|P'_1|+|P'_2|)+1= k(SC(v)-1)+1\le kSC(v)$.
\end{proof}
We then use spanners instead of fault-tolerant spanners to get better multiplicative approximation factor. The next two lemmas help us approximate shortest cycles using spanners.  Let $H_{\alpha,\beta}$ be an $(\alpha,\beta)$-spanner of the graph.
\begin{lemma}
\label{lem:cycle-spanner-general}
If $SC(v)>\ceil{\alpha+2}(\alpha+\beta+1)$, then there is a cycle of length at most $\alpha SC(v)+\ceil{\alpha+2}\beta$ in $H_{\alpha,\beta}$.
\end{lemma}
\begin{proof}
Divide $C_v$ into $a=\ceil{\alpha+2}$ sections of almost the same size, where $u_1,\ldots,u_{a-1}\in C_v$ form the subdivision. Let $v=u_0=u_a$. We have that the length of the subpath $u_iu_{i+1}$ for $i=1,\ldots,a-1$ is either $\floor{\frac{SC(v)}{a}}$ or $\floor{\frac{SC(v)}{a}}+1$ and is equal to $d_G(u_i,u_{i+1})$.

We show that the union of subpaths $\pi_{H_{\alpha,\beta}}(u_i,u_{i+1})$ for $i=0,\ldots,a-1$ contains a cycle around $v$. Then this cycle will be of length at most $\sum_{i=0}^{a-1}d_{H_{\alpha,\beta}}(u_i,u_{i+1})\le\sum_{i=0}^{a-1}(\alpha d_{G}(u_i,u_{i+1})+\beta) =\alpha SC(v)+a\beta$.

First note that 
\begin{align*}
d_{H_{\alpha,\beta}}(u_i,u_{i+1})+d_{G}(u_i,u_{i+1})&\le (\alpha+1)d_{G}(u_i,u_{i+1})+\beta\\
&\le \frac{\alpha+1}{a}SC(v)+\alpha+1+\beta\\
&<\frac{\alpha+1}{a}SC(v) + \frac{SC(v)}{a} =\frac{\alpha+2}{a}SC(v)\le SC(v).
\end{align*}
Where the third inequality comes from the fact that $SC(v)>a(\alpha+\beta+1)$

Now we show that the edge attached to $v$ in $\pi_G(v,u_1)$ and $\pi_{H_{\alpha,\beta}}(v,u_1)$ are the same. If not, then the union of these two paths contains a cycle around $v$, and by the above inequality the size of this cycle is less than $SC(v)$. Similarly, the edge attached to $v$ in $\pi_G(v,u_{a-1})$ and $\pi_{H_{\alpha,\beta}}(v,u_{a-1})$ are the same.

Finally, for $i=1,\ldots,a-2$, the path $\pi_{H_{\alpha,\beta}}(u_i,u_{i+1})$ does not pass $v$. Because otherwise, the union of $\pi_{H_{\alpha,\beta}}(u_i,u_{i+1})$ and $\pi_{G}(u_i,u_{i+1})$ contain a cycle around $v$ which is of length less than $SC(v)$. So the union of subpaths $\pi_{H_{\alpha,\beta}}(u_i,u_{i+1})$ for $i=0,\ldots,a-1$ contains a cycle around $v$. 
\end{proof}

\begin{lemma}
\label{lem:spanner-u}
Let $u\in C_v$ and suppose $SC(v)>\ceil{\alpha+2}(\alpha+\beta+1)$. There are two sets of paths between $u$ and $v$ in $H_{\alpha,\beta}$ that contain a cycle around $v$ and the sum of their sizes is at most $\alpha SC(v)+\ceil{\alpha+3}\beta$.
\end{lemma}

\begin{proof}
Let $a=\ceil{\alpha+2}$. Consider the division of $C_v$ into $a$ sections with almost equal sizes. Adding $u$ to this subdivision, we have $a+1$ sections of length at most $\floor{\frac{SC(v)}{a}}+1$. Let the nodes on the cycle that form these sections be $u_1,\ldots,u_{a}$, and let $v=u_0=u_{a+1}$. Note that for some $i$, $u_i=u$. The union of $\pi_{H_{\alpha,\beta}}(u_j,u_{j+1})$ for $j=0,\ldots,i-1$ includes one path from $v$ to $u$, and the union of $\pi_{H_{\alpha,\beta}}(u_j,u_{j+1})$ for $j=i,\ldots,a$ includes the other path from $u$ to $v$. The rest of the proof is similar to Lemma \ref{lem:cycle-spanner-general}. The union of $\pi_{H_{\alpha,\beta}}(u_i,u_{i+1})$ for $i=0,\ldots,a$ contains a cycle around $v$, and the length of the cycle is at most the sum of the lengths of the paths, which is at most $\alpha SC(v)+(a+1)\beta$.
\end{proof}
Finally, we use the following lemma to compute small cycles. The proof is in the appendix.
\begin{restatable}{lemma}{smallcycles}
\label{thm:4k2^2k-approx}
Let $k\ge 2$ be a fixed integer. There is an algorithm for unweighted ANSC that runs in $\tilde{O}(n^{1+2/k}+m)$ time and for any node $v\in V(G)$ gives an estimate of $(2k-1)2^{k-2}SC(v)+2^{k-1}$.
\end{restatable}

Our first algorithm is as follows.
\sixoneanscundir*

\begin{proof}
We call $u$ \textit{low} degree if $d_G(u)\le \sqrt{n} $, and \textit{high} degree otherwise. The algorithm consists of two parts. 

\paragraph{Part 1.} We remove every edge from the graph that has two high degree endpoints, and call the remaining graph $G'$.
So we end up with at most $n^{1.5}$ edges in $G'$. Run the $3+\epsilon$-approximation algorithm for ANSC of Theorem \ref{thm:kapprox-ansc-undir} on $G'$ in $\tilde{O}(n^{1.5}n^{1/3})=\tilde{O}(n^{1+5/6})$ time.
If $u$ is node that its shortest cycle $C_u$ doesn't have two consecutive high degree nodes, then $C_u$ is preserved in $G'$ and this step gives an estimate of at most $(3+\epsilon)SC(u)\le 6SC(u)$ for the shortest cycle passing through $u$.

\paragraph{Part 2.} Let $P$ be a $1$-fault-tolerant $5$-spanner of $G$. We can build such a spanner of size $O(n^{4/3})$ in $\tilde{O}(m)$ time by Lemma \ref{lem:fault-tolerant-const}. Let $S$ be a set of samples of size $\tilde{O}(n^{1/2})$, which hits the neighborhood of every high degree node in $G$ with high probability.
We do Dijkstra in $G$ from the set $S$ in $\tilde O(m)$ time and so for each node $v$, we can find $p(v)\in S$ that is the closest node in $S$ to $v$, breaking ties arbitrarily. Note that since we did Dijkstra from the set $S$ to find $p(v)$s, for any node $u$ on the $vp(v)$ path we have $p(u)=p(v)$. 
For each node $s\in S$, let $F(s)$ be the set of nodes $v$ such that $s=p(v)$. Note that $v$ can be high or low degree. Now for each $s\in S$, let $G_s$ be the graph consisting of $P$ and all the edges adjacent to a node in $F(s)$ in $G$. So if $v\in F(s)$, all the edges adjacent to $v$ in $G$ are present in $G_s$. For each $s\in S$, we do  \textsc{CycleEstimationDijkstra}($s$) in $G_s$ in time $\tilde O(m_{G_s})$. Note that if $e$ is an edge that is not in $P$, it appears in at most $2$ graphs $G_s$. So all these Dijkstras together take $\sum_s \tilde O(m_{G_s})\le \tilde O(m+|S|\cdot |E_P|) \le \tilde{O}(m+n^{1/2}\cdot n^{4/3})\le \tilde{O}(m+n^{1+5/6})$ time.  

Now we prove that we have a good estimate for all nodes $v$. Suppose $v$ is a high degree node. So $p(v)$ is a neighbor of $v$. By Lemma \ref{lem:cycle-est-fault-tolerant} there is a cycle of length at most $5SC(v)$ passing through $v$ in $G_{p(v)}$. So 
\textsc{CycleEstimationDijkstra}($p(v)$)
gives an estimate of at most $5SC(v)+2$ for the length of the smallest cycle passing through $v$.

Now suppose that $v$ is a low degree node such that on its shortest cycle in $G$ there are two consecutive high degree nodes. So there exist a high degree node $w$ on $C_v$ such that $d_G(v,w)\le \floor{\frac{SC(v)-1}{2}}$. This high degree node must be adjacent to a node $s$ in $S$, so $d_G(s,v)\le d_G(v,w)+1\le \floor{\frac{SC(v)+1}{2}}$. This means that $d_G(v,p(v))\le \floor{\frac{SC(v)+1}{2}}$. Moreover, for any node $u$ on the $vp(v)$ path in $G$, $p(v)$ must be the closest node in $S$ to $u$ (i.e., $p(u)=p(v)$), so the $vp(v)$ path  exists in $G_{p(v)}$.
By Lemma \ref{lem:cycle-est-fault-tolerant}, there is a cycle of length at most $5SC(v)$ passing through $v$ in $G_{p(v)}$. So  \textsc{CycleEstimationDijkstra}$(p(v))$ gives an estimate of $5SC(v)+2\floor{\frac{SC(v)+1}{2}}\le 6SC(v)+1$ for $SC(v)$ by Lemma \ref{lem:datastructure}. 

Note that by part 1 if there are no such high degree nodes on the shortest cycle of $v$, we already have a good estimate for $SC(v)$.
\end{proof}


The structure of the next algorithm is the same as Theorem \ref{thm:6-1anscundir}, however we use don't use fault-tolerant spanners.
\ansctwoapproxbest*
\begin{proof}
Let $\epsilon' = \epsilon/2$.We call $u$ \emph{low} degree if $d_G(u)\le \sqrt{n}$, and \emph{high} degree otherwise. 
\paragraph{Part 1.} We remove every edge with two high degree endpoints, and call the remaining graph $G'$. Note that $G'$ has at most $O(n^{1.5})$ edges. We compute the $(1+\epsilon', \beta')$-spanner of Corollary \ref{cor:very-sparse-spanner}, which has $O(n^{1+\epsilon'})$ edges and is computable in $\tilde{O}(n^{1.5+\epsilon'})$ time, for some $\beta'$ that is a function of $\epsilon'$. Call this spanner $P$. We run the $2$-approximation ANSC algorithm of Theorem \ref{thm:warmupp} on this spanner which runs in $\tilde{O}(n^{1.5+\eps'})$ time. This algorithm gives our first estimate for $SC(v)$ for all $v$.
\paragraph{Part 2.} Let $S$ be a set of samples of size $\tilde{O}(\sqrt{n})$, which hits the neighborhood of every high degree node in $G$ with high probability. Compute a $(2+\eps,\beta'')$-spanner $H$ of Lemma \ref{lem:lineartime-sparse} on the input graph, which has $O(n^{1+\eps})$ edges and can be computed in $\tilde{O}(m + n^{1.5+\eps})$ time, for some $\beta''$ that is a function of $\eps$. Run \textsc{CycleEstimationDijkstra}($s$) for every $s\in S$. Each \textsc{CycleEstimationDijkstra} gives an estimation for $SC(v)$ for every $v$.
\paragraph{Part 3.} Run the algorithm of Lemma \ref{thm:4k2^2k-approx} for $k=4$ in time $\tilde{O}(n^{1.5}+m)$, to get a $(28,8)$-approximation for $SC(v)$ for every $v$. 

The final estimation for $SC(v)$ for every $v$ is the minimum of all of the estimations. 
Now we prove the correctness of the algorithm.

Consider a node $v$. Let $\alpha = 2+2\eps'=2+\eps$. If $SC(v)\le\ceil{\alpha+2}(\alpha+\max(\beta',\beta'' )+1)$, Part 3 gives a $(28,8)$-approximation of $SC(v)$, so assuming $\beta\ge28\ceil{\alpha+2}(\alpha+\max(\beta',\beta'' )+1)+8$, our algorithm gives the desired approximation guarantees. so suppose that $SC(v)>\ceil{\alpha+2}(\alpha+\max(\beta',\beta'' )+1)$.

Suppose that the shortest cycle passing through $v$,$C_v$, is fully contained in $G'$. Then by Lemma \ref{lem:cycle-spanner-general}, there is cycle of length $(1+\eps')SC(v)+\ceil{2+\eps'}\beta' $ in the spanner in Part 1. So the $2$-approximation algorithm gives a $(2+\eps,2\ceil{2+\eps'}\beta' )$ approximation for $SC(v)$. 

Now suppose that $C_v$ is not fully contained in $G'$. So there is a high degree vertex $u$ in $C_v$, and a $s\in S$ that is a neighbor of $u$. By Lemma \ref{lem:spanner-u}, there is a cycle of length $(2+\epsilon)SC(v)+\ceil{5+\epsilon}\beta'' $ around $v$ that passes through $u$ in $H$. So \textsc{CycleEstimationDijkstra}($s$) gives an estimate of $(2+\epsilon)SC(v)+\ceil{5+\epsilon}\beta'' +2$ for $SC(v)$ by Lemma \ref{lem:datastructure}. 
Finally set $\beta = \max(28\ceil{\alpha+2}(\alpha+\max(\beta',\beta'' )+1)+8,2\ceil{2+\eps'}\beta',\ceil{5+\epsilon}\beta'' +2)$.
\end{proof}

%% file: appendix.tex
\appendix

\section{Reductions between $n$-PSP and ANSC}

\obs*
\begin{proof}
Given a weighted undirected $n$-PSP instance, let $M$ denote the maximum edge weight in the input graph $G$. Then, create a new graph by adding a new vertex $v_i$ to the graph for every vertex pair $(s_i,t_i)$ being queried, and adding edges $(v_i,s_i),(s_i,t_i)$ of weights $10 nM$.
 These added edges have weights larger than any simple path in the original graph, so these added nodes and edges do not change the shortest path from $s_j,t_j$ for any $j$.
  It is then obvious that the shortest cycle passing through the new vertex $v_i$ must consist of the two added edges and the shortest path between $s_i$ and $t_i$ in the input graph $G$, with cycle length $20nM+ d_G(s_i,t_i)$. Then we can find $d_G(s_i,t_i)$ for all $i$ by running the ANSC algorithm.
\end{proof}
\obss*
\begin{proof}
Given a directed  ANSC instance, we create a new graph $G'$ where each vertex $v$ in the input graph $G$ has two copies $v_1,v_2$ in $G'$. For every input edge $(u,v)$, we add edges $(u_i,v_j)$ for $(i,j)\in \{1,2\}\times \{1,2\}$ with the same weight in $G'$. Then, the distance from $v_1$ to $v_2$ in $G'$ equals the shortest cycle in $G$ passing through $v$. This is because for each path from $v_1$ to $v_2$ there is cycle around $v$ in $G$ of the same length, and for every cycle around $v$, there is a path from $v_1$ to $v_2$ of the same length. 
\end{proof}

\section{Hardness Assumptions}\label{app:assump}
Here we give a complementary list of hardness assumptions that we use for proving our conditional lower bounds.

For sparse graphs the best known algorithm for triangle detection is a 25-year-old algorithm that runs in time $O(m^{1.41})$ \cite{alon1997finding} using the current best bound on $\omega$ \cite{alman2021refined}, and in time $\tilde{O}(m^{4/3})$ time if $\omega=2$. Hardness results based on reductions from triangle detection are very standard (including before the explicit introduction of fine-grained complexity e.g. \cite{DBLP:conf/stoc/Chan02,chan2006dynamic}). Like \cite{abboud2014popular} and many others, we hypothesize that there is no near-linear time algorithm for triangle detection.

\begin{hypothesis}[Sparse Triangle Hypothesis] \label{hyp:sparsetri}
In the word-RAM model with $O(\log n)$ bit words,
there exists a constant $\delta>0$ such that triangle detection requires time $m^{1+\delta-o(1)}$ in undirected graphs.
\end{hypothesis}

For dense graphs, the best known algorithm for triangle detection runs in time $O(n^{\omega})$ and is a simple application of matrix multiplication. It is a big open problem whether there is a faster algorithm (see \cite{williams2010subcubic}, \cite{ woeginger2008open} Open Problem 4.3(c), \cite{spinrad2003efficient} Open Problem 8.1), and it is generally believed that there is not. The Dense Triangle Hypothesis has also been used as a hardness assumption in \cite{bergamaschi2021new}.

\begin{hypothesis}[Dense Triangle Hypothesis] \label{hyp:densetri}
In the word-RAM model with $O(\log n)$ bit words,
for any constant $\eps > 0$, there is no $O(n^{\omega-\eps})$ time algorithm for triangle detection in undirected graphs.
\end{hypothesis}

Our next hypothesis concerns the problem of detecting a \emph{simplicial vertex}, a vertex whose neighborhood induces a clique (in an undirected graph). The best known algorithm for detecting a simplicial vertex runs in time $O(n^{\omega})$ \cite{kloks2000finding}. It appears on Spinrad's list of open problems \cite{Spi} and has been used as a hard problem in reductions by Kratsch and Spinrad \cite{DBLP:journals/siamcomp/KratschS06}.

\begin{hypothesis}[Simplicial Vertex Hypothesis] \label{hyp:simp}
In the word-RAM model with $O(\log n)$ bit words,
for any constant $\eps > 0$, there is no $O(n^{\omega-\eps})$ time algorithm for detecting a simplicial vertex.
\end{hypothesis}

Our next hypothesis concerns the problem of detecting a $k$-cycle in \emph{directed} graphs, and has been used as a hardness hypothesis in \cite{ancona2018algorithms, lincolnsoda18,probst2020new, girth-icalp}.

\begin{hypothesis}[$k$-Cycle Hypothesis] \label{hyp:cycle}
In the word-RAM model with $O(\log n)$ bit words,
for any constant $\eps > 0$, there exists a constant integer $k$, so that there is no $O(m^{2-\eps})$ time algorithm for $k$-cycle detection in directed graphs.
\end{hypothesis}

Our final hypothesis concerns the All-Edges Sparse Triangle problem. In this problem, we are given an undirected graph of $m$ edges, and for each edge $e$ the goal is to output whether $e$ is in a triangle. 
The following All-Edges Sparse Triangle Hypothesis is implied by two of the most standard hypotheses: the 3-SUM hypothesis \cite{DBLP:conf/stoc/Patrascu10} as well as the APSP hypothesis \cite{DBLP:conf/focs/WilliamsX20}.
\begin{hypothesis}[All-edges Sparse Triangle Hypothesis] \label{hyp:all-edge-sparse}
There is no $O(m^{4/3-\eps})$ time algorithm for the All-Edges Sparse Triangle problem on $m$-edge graphs, for any constant $\eps>0$.
\end{hypothesis}


\section{Conditional Lower Bounds}\label{app:lb}

\subsection{$n$-PSP}
We begin with a simple observation, which follows from the fact that a better than $4/3$-approximation for the girth returns the value $3$ if and only if the graph contains a triangle.

\begin{observation}\label{obs:lb}
Under the Sparse Triangle Hypothesis, any better than $4/3$-approximation for Girth in undirected graphs requires time $m^{1+\delta-o(1)}$. Under the Dense Triangle Hypothesis, any better than $4/3$-approximation for Girth in undirected graphs requires $n^{\omega-o(1)}$ time. 
\end{observation}

Now we provide a reduction from Triangle Detection to $n$-Pairs Diameter as stated in \cref{thm:tripair}.

\begin{theorem}\label{thm:tripair}
Under the Sparse Triangle Hypothesis, getting a better than $5/3$-approximation for $n$-Pairs Minimum Distance in undirected graphs requires time $m^{1+\delta-o(1)}$. Under the Dense Triangle Hypothesis, any better than $5/3$-approximation for $n$-Pairs Minimum Distance in undirected graphs requires time $n^{\omega-o(1)}$. 
\end{theorem}

\begin{proof}
Let $G=(V,E)$ be a instance of Triangle Detection. The reduction is simple. We construct a 4-layered graph $G'$ with layers $V_1,V_2,V_3,V_4$ where $V_i$ is a copy of $V$, and $v_i\in V_i$ is a copy of $v\in V$ for all $i=1,2,3,4$. For all $(u,v)\in E$, we add an edge between $u_1$ and $v_2$, an edge between $u_2$ and $v_3$, and an edge between $u_3$ and $v_4$. The set of pairs $(s_i,t_i)$ for our $n$-Pairs Diameter instance is $(v_1,v_4)$ for every vertex $v\in V$.

We claim that if $G$ contains a triangle, then the $n$-pairs minimum distance in $G'$ is 3, and if $G$ does not contain a triangle, then the $n$-pairs minimum distance in $G'$ is at least 5.

First suppose that $G$ contains a triangle $(u,v,w)$. Then, $(u_1,v_2)$, $(v_2,w_3)$, and $(w_3,u_4)$ are all edges in $G'$. Thus, $d(u_1,u_4)=3$.

Now, suppose that $G$ does not contain a triangle. Let $(v_1,v_4)$ be the pair that realizes the $n$-pairs minimum distance. Since $G'$ is bipartite with $V_1$ and $V_4$ on opposite sides of the bipartition, we know that $d(v_1,v_4)$ is odd. Thus, it suffices to show that $d(v_1,v_4)\not=3$. Suppose for contradiction that $d(v_1,v_4)=3$ and let $u_2\in V_2$ and $w_3\in V_3$ be vertices on a path of length 3 between $v_1$ and $v_4$. Then, $(v,u)$, $(u,w)$, and $(w,v)$ are all edges in $G$, a contradiction.
\end{proof}

Next, we provide a reduction from Simplicial Vertex to $n$-Pairs Diameter as stated in \cref{thm:simppair}.

\begin{theorem}\label{thm:simppair}
Under the Simplicial Vertex Hypothesis, any better than $5/3$-approximation algorithm for $n$-Pairs Diameter in undirected graphs requires time $n^{\omega-o(1)}$. 
\end{theorem}

\begin{proof}

Let $G=(V,E)$ be an instance of Simplicial Vertex.

 We construct a 4-layered graph $G'$ with layers $V_1,V_2,V_3,V_4$. Each vertex $v\in V$ has a copy $v_i$ in each $V_i$. For all $(u,v)\in E$, we add an edge between $u_1$ and $v_2$, and an edge between $u_3$ and $v_4$. For each pair $(u,v)$ of distinct vertices which are \emph{not} adjacent in $G$, we add an edge between $u_2$ and $v_3$. The set of pairs $(s_i,t_i)$ for our $n$-Pairs Diameter instance is $(v_1,v_4)$ for every vertex $v\in V$.

We claim that if $G$ contains a simplicial vertex, then the $n$-pairs diameter of $G'$ is at least 5, and if $G$ does not contain a simplicial vertex, then the $n$-pairs diameter of $G'$ is 3.

First, suppose $G$ contains a simplicial vertex $v$. We claim that $d(v_1,v_4)\geq 5$. First, since $G'$ is bipartite with $V_1$ and $V_4$ on opposite sides of the bipartition, we know that $d(v_1,v_4)$ is odd. Thus, it suffices to show that $d(v_1,v_4)\not=3$. Suppose for contradiction that $d(v_1,v_4)=3$ and let $u_2\in V_2$ and $w_3\in V_3$ be vertices on a path of length 3 between $v_1$ and $v_4$. We note that $u$, $v$, and $w$ are distinct vertices in $G$ since edges in $G'$ only occur between copies of distinct vertices in $G$. By the construction of $G'$, $(v,u)$ and $(v,w)$ are both edges in $G$, while $(u,w)$ is not. Thus by definition, $v$ is not a simplicial vertex, a contradiction.

Now, suppose that $G$ does not contain a simplicial vertex. We will show that the $n$-pairs diameter of $G$ is 3. Consider an arbitrary vertex $v_1\in V_1$. Since $v$ is not a simplicial vertex in $G$, $v$ has two neighbors $u$ and $w$ such that $(u,w)\not\in E$. Thus, $(v_1,u_2)$, $(u_2,w_3)$, and $(w_3,v_4)$ are all edges in $G'$, so $d(v_1,v_4)=3$. 
\end{proof}

\begin{theorem}\label{thm:kcycle}
Under the $k$-Cycle Hypothesis, no $m^{2-\eps}$ time algorithm for any constant $\eps>0$ can provide any finite approximation for $n$-Pairs Minimum Distance in directed graphs. 
\end{theorem}

\begin{proof}
Let $G=(V,E)$ be an instance of $k$-Cycle.  First we use
color-coding \cite{colorcoding} so that with polylogarithmic time overhead, we can assume that the vertices of $G$ are partitioned into $V_1, V_2, \dots, V_k$, so that if $G$ contains a $k$-cycle, $G$ has a $k$-cycle whose $i^{th}$ vertex is in $V_i$ for all $i$. We consider the subgraph of $G$ containing only edges from $V_i$ to $V_{i+1}$ for some $i$, and edges from $V_k$ to $V_1$. We construct a new set of vertices $V_{k+1}$ which contains a copy of each vertex in $V_1$, such that $v\in V_1$ and its copy $v'\in V_{k+1}$ have identical in- and out-neighborhoods. Now, we remove all in-edges to $V_1$, and all out-edges from $V_{k+1}$. Let $G'$ be the resulting graph. Note that $G'$ only contains edges from $V_i$ to $V_{i+1}$ for some $i$. The set of pairs $(s_i,t_i)$ for our $n$-Pairs Minimum Distance instance consists of each vertex in $V_1$ paired with its copy in $V_{k+1}$.

We claim that if $G$ contains a $k$-cycle, then the $n$-pairs minimum distance in $G'$ is finite, and if $G$ does not contain a $k$-cycle, then the $n$-pairs minimum distance in $G'$ is infinite.

First, suppose that $G$ contains a $k$-cycle on vertices $v_i\in V_i$ for all $1\leq i\leq k$. Then there is a path in $G'$ from $v_1\in V_1$ to $v_k\in V_k$ through each $v_i\in V_i$, and an edge from $v_k$ to the copy $v'_1$ of $v_1$ in $V_{k+1}$. Thus, $d(v_1,v'_1)$ is finite so the $n$-pairs minimum distance in $G'$ is finite.

Now, suppose that the $n$-pairs minimum distance in $G'$ is finite. Then, there exists a vertex $v\in V_1$ that can reach its copy $v'\in V_{k+1}$. Since all edges are directed from $V_i$ to $V_{i+1}$ for some $i$, there is a path of length $k$ from $v$ to $v'$ that goes through every layer from $V_1$ to $V_{k+1}$ in order exactly once. By identifying $v$ and $v'$ we obtain a $k$-cycle in $G$.
\end{proof}

\subsection{ANSC}

\begin{theorem}\label{thm:aest}
Under the All-Edges Sparse Triangle Hypothesis, any better than $3/2$-approximation algorithm for ANSC requires time $m^{4/3-o(1)}$. 
\end{theorem}
\begin{proof}
Given an All-Edges Sparse Triangle instance $G=(V,E)$ with $m$ edges, we make three copies $V_1,V_2,V_3$ of the vertex set $V$, and between every pair of $V_i,V_j$ ($(i,j)\in \{(1,2),(2,3),(3,1)\}$) we copy the edge set $E$. It suffices to solve All-Edges Sparse Triangle on this tripartite graph. 

Then, we subdivide every edge $e=(v_1,v_2)$ between $V_1,V_2$, i.e., add a vertex $u_e$ and replace this edge by two edges $(v_1,u_e),(u_e,v_2)$. 

In this new 4-partite graph (which has $4m$ edges), node $u_e$ is in a 4-cycle $(u_e\to v_2 \to v_3 \to v_1\to u_e)$ if and only if the original edge $e=(v_1,v_2)$ is in a triangle $(v_1,v_2,v_3)$. Otherwise, the shortest cycle passing through $u_e$ must have length at least $6$. Hence, a better than $3/2$ approximation algorithm that solves ANSC on this new 4-partite graph can be used to solve All-Edges Sparse Triangle on the original graph, which requires time $m^{4/3-o(1)}$ under the All-Edges Sparse Triangle hypothesis.
\end{proof}

Now we provide a reduction from Simplicial Vertex to Cycle Diameter as stated in \cref{thm:simpcycle}.



\begin{theorem}\label{thm:simpcycle}
Under the Simplicial Vertex Hypothesis, any better than $7/5$-approximation algorithm for Cycle Diameter in undirected graphs requires time $n^{\omega-o(1)}$. 
\end{theorem}

\begin{proof}

Let $G=(V,E)$ be an instance of Simplicial Vertex.

 We construct a 5-partite graph $G'$ with 5-partition $V_1,V_2,V_3,V_4,V_5$. Each vertex $v\in V$ has a copy $v_i$ in each $V_i$. 
 For all $(u,v)\in E$, we add an edge between $u_4$ and $v_5$, and an edge between $u_2$ and $v_1$. For each pair $(u,v)$ of distinct vertices which are \emph{not} adjacent in $G$, we add an edge between $u_1$ and $v_5$. For every $v\in V$, add edges $(v_2,v_3)$ and $(v_3,v_4)$.
 
 Additionally, add four dummy vertices $z_1,z_2,z_3$ and $z_4$, where both $z_1,z_2$ are adjacent to every vertex in $V_1$, and both $z_3,z_4$ are adjacent to every vertex in $V_5$. There is also an edge between $z_1$ and $z_2$, and an edge between $z_3$ and $z_4$.

We claim that if $G$ contains a simplicial vertex, then the cycle diameter of $G'$ is at least 7, and if $G$ does not contain a simplicial vertex, then the cycle diameter of $G'$ is at most 5.

First, suppose $G$ contains a simplicial vertex $v$. We claim that the length of the shortest cycle through $v_3$ is at least 7. Observe that $\deg(v_3)=2$, and hence the shortest cycle through $v_3$ must have length $d+2$, where $d$ denotes the length of the shortest path from $v_2$ to $v_4$ (which are the only neighbors of $v_3$) without going through $v_3$. Suppose for contradiction that the length $d$ of this path from $v_2$ to $v_4$ is at most 4. Then, by inspecting the structure of the graph $G'$, we observe that the only possibility of this path is to go through
$v_2 \to V_1 \to V_5 \to v_4$. 
Let $u_1\in V_1$ and $w_5\in V_5$ be the other vertices on this path. We note that $u$, $v$, and $w$ are distinct vertices in $G$ since edges in $G'$ only occur between copies of distinct vertices in $G$. By the construction of $G'$, $(v,u)$ and $(v,w)$ are both edges in $G$, while $(u,w)$ is not. Thus by definition, $v$ is not a simplicial vertex, a contradiction.

Now, suppose that $G$ does not contain a simplicial vertex. We will show that the cycle diameter of $G$ is at most 5. We first note that every vertex in $V_1\cup V_5$ participates in a triangle with $z_1$ and $z_2$ or $z_3$ and $z_4$. Thus, it remains to show that every vertex in $V_2\cup V_3\cup V_4$ participates in a 5-cycle. 
Consider an arbitrary vertex $v\in V$. Since $v$ is not a simplicial vertex in $G$, $v$ has two neighbors $u$ and $w$ such that $(u,w)\not\in E$. Thus, $(v_3,v_2),(v_2,u_1),(u_1,w_5),(w_5,v_4)$, and $(v_4,v_3)$, are all edges in $G'$, so $v_2,v_3,v_4$ participate in a 5-cycle.
\end{proof}

Next, we provide an unconditional lower bound for ANSC as stated in \cref{thm:an4c}.

\begin{theorem}\label{thm:an4c}
Any better than $3/2$-approximation algorithm for ANSC on graphs with $m=\Theta(n^2)$ (where the input graph is represented as an adjacency matrix) requires $\Omega(n^2)$ time unconditionally. 
\end{theorem}

\begin{proof}
We consider the Disjointness problem on $n^2$-length vectors: given two vectors $x,y\in \{0,1\}^{n^2}$, answer whether $\sum_{j=1}^{n^2} x[j]\cdot y[j] \ge 1$. It is well-known that the query complexity of this problem is $\Omega(n^2)$.

We can solve this problem using a better than $3/2$ approximation algorithm for ANSC as follows. First, write $x\in \{0,1\}^{n^2}$ as the concatenation of $n$ vectors $x_1,x_2,\dots,x_n \in \{0,1\}^n$, and similarly write $y\in \{0,1\}^{n^2}$ as the concatenation of $n$ vectors $y_1,y_2,\dots,y_n \in \{0,1\}^n$. Then, $\sum_{j=1}^{n^2} x[j]\cdot y[j] \ge 1$ if and only if there exists $i$ such that $\sum_{j=1}^{n} x_i[j]\cdot y_i[j] \ge 1$.

We create $n$ nodes $a_1,a_2,\dots,a_n$  that are used to encode the $n$ coordinates. For every pair $x_i,y_i\in \{0,1\}^n$ of vectors, create three nodes $u_i,v_i,w_i$, and connect edges $(u_i,w_i),(v_i,w_i)$, and for every $1\le j\le n$, connect  $(u_i,a_j)$ if and only if $x_i[j]=1$, and connect $(v_i,a_j)$ if and only if $y_i[j]=1$. Then, observe that $\sum_{j=1}^n x_i[j]\cdot y_i[j] \ge 1$ if and only if the node $w_i$ is in a 4-cycle $w_i\to v_i\to a_j\to u_i \to w_i$, where $x_i[j]=y_i[j]=1$. Since it is a bipartite graph, the shortest cycle passing through $w_i$ is either 4 or at least $6$. Hence, we can solve the Disjointness problem by solving ANSC (with better than $3/2$ approximation) on this $O(n)$-vertex graph. This implies that we need $\Omega(n^2)$ query complexity to solve ANSC with better than $3/2$ approximation, when the input graph is represented as an adjacency matrix.
\end{proof}

\section{Approximation Algorithms for $n$-PSP}

\begin{theorem}
\label{thm:pd-k-tz}
Given an $n$-node $m$-edge undirected weighted graph $G$ and vertex pairs $(s_i,t_i)$ for $1\leq i\leq O(n)$, there is a randomized algorithm for $n$-PSP that computes an estimate $\hat{d}(s_i,t_i)$ such that $d(s_i,t_i)\leq \hat{d}(s_i,t_i) \leq (2k-3)d(s_i,t_i)+2\lceil{d(s_i,t_i)/2}\rceil$ in $\tilde{O}(mn^{1/k})$ time with high probability.
\end{theorem}

\begin{proof}[Proof sketch]
The construction and proof are very similar to Thorup-Zwick distance oracles \cite{thorup2005approximate}. The only difference is a modification of the base case. We only outline the differences here. 

We perform the entire preprocessing of \cite{thorup2005approximate}. Then for each input pair $(s_i,t_i)$, we perform the query algorithm of \cite{thorup2005approximate} with the following change to the base case. The base case in \cite{thorup2005approximate} is ``if $t_i\in B(s_i)$, then return $d(s_i,t_i)$'' (where $B(s_i)$ is the ``bunch'' of $s_i$ defined in \cite{thorup2005approximate}). Because there are only $O(n)$ pairs  $(s_i,t_i)$, we can afford to compute whether there is a vertex $v\in B(s_i)\cap B(t_i)$. If so, we return $\min_{v\in B(s_i)\cap B(t_i)} d(s_i,v)+d(v,t_i)$. If there is such a $v$, then there is such a $v$ on a shortest path between $s_i$ and $t_i$, so we return exactly $d(s_i,t_i)$. If $B(s_i)\cap B(t_i)=\emptyset$, then either $d(s_i,p_1(s_i))\leq \lceil{d(s_i,t_i)/2}\rceil$ or $d(t_i,p_1(t_i))\leq \lceil{d(s_i,t_i)/2}\rceil$ (where $p_1(s_i)$ and $p_1(t_i)$ are defined in \cite{thorup2005approximate}). Suppose the latter without loss of generality. Then, by the triangle inequality, $d(s_i,p_1(t_i))\leq d(s_i,t_i)+\lceil{d(s_i,t_i)/2}\rceil$. So, $d(s_i,p_1(t_i))+d(p_1(t_i),t_i)\leq d(s_i,t_i)+2\lceil{d(s_i,t_i)/2}\rceil$. 

Completing the analysis of the query algorithm as in \cite{thorup2005approximate}, we get a final estimate of at most $(2k-3)d(s_i,t_i)+2\lceil{d(s_i,t_i)/2}\rceil$.
\end{proof}

\begin{theorem}
\label{thm:pd-ksquared-2ktime}
Given an $n$-node $m$-edge undirected weighted graph $G$ and vertex pairs $(s_i,t_i)$ for $1\leq i\leq O(n)$, there is a randomized algorithm that computes a $((2k-1)\cdot (2k-2))$-approximation for $n$-PSP in $\tilde{O}(n^{1+2/k}+m)$ time with high probability.
\end{theorem}

\begin{proof}
By \cite{2k-1spanner,DBLP:conf/icalp/RodittyTZ05}, we can build a $(2k-1)$-spanner $P$ of the input graph in $O(km)$ time with only $m' = O(k n^{1+1/k})$ edges. Then, we run the $(2k-2)$-approximate $n$-PSP algorithm from \cref{thm:pd-k-tz} on $P$, in time $\tilde O(m'n^{1/k}) \le \tilde O(n^{1+2/k})$. For every vertex pair $(u,v)$, the estimate for $d_G(u,v)$ is at most $(2k-2)\cdot d_P(u,v) \le (2k-2)\cdot (2k-1)\cdot d_G(u,v)$. The overall running time is $\tilde O(m+n^{1+2/k})$.
\end{proof}

Theorem \ref{thm:pd-ksquared-2ktime} is comparable to the result of Wulff-Nilsen \cite{apsp1}. Their result can be written as a $(4k^2-1)$-approximation DO with preprocessing time $O(km+n^{1+\frac{c}{k\sqrt{2}}})$ for some constant $c$ and query time $O(k^2)$, which gives a $(4k^2-1)$ $n$-PSP algorithm in time $O(km+n^{1+\frac{c}{k\sqrt{2}}})$. The value of $c$ is rather large (they set it to be $9+3\sqrt{13}$) and so our algorithm has a faster running time especially for smaller values of $k$. They also give a DO with faster preprocessing time for small $k$. For example, for $k\ge 3$, their $(2k-1)$-approximation DO has preprocessing time $O(km+kn^{3/2+2/k})$ (the preprocessing time is slightly different if $k$ is not $0$ mod $3$). This translates to a $(4k^2-1)$-approximation algorithm for $n$-PSP in time $O(m+n^{3/2+1/k^2})$, which is still slower than the bound of Theorem \ref{thm:pd-ksquared-2ktime} and has a slightly worse approximation factor. 

\section{Approximation Algorithms for Directed ANSC}

We are going to generalize the results on the girth to ANSC problem. 
More specifically, we tweak the result of \cite{girth-icalp} that gives a $(2k+\epsilon)$-approximation for the girth to a $(2k+1+\epsilon)$-approximation for ANSC. The formal statement of the result is stated below. For $k=2$, we improve the approximation factor and give a $(2+\epsilon)$-approximation, the same approximation factor as the girth.

\begin{theorem}\label{thm:2k+1ansc-approx}
Given an $n$-node $m$-edge directed graph $G$ with edge weights in $\{1,\ldots,M\}$, a constant $\epsilon>0$, and an integer $k\ge 1$ , there is a randomized $(2k+1+\epsilon)$-approximation algorithm for ANSC
 that runs in $\tilde{O}(mn^{\alpha_k}\log{(M)})$ time with high probability, where $\alpha_k>0$ is the solution to $\alpha_k(1+\alpha_k)^{k-1}=1-\alpha_k$.
\end{theorem}



We give a few definitions that help us prove the above theorem. 
Define for each $j=0,\ldots,\log_{1+\epsilon}Mn$ the level $j$, $B^j(u)$, and the ball $j$ around $u$, $\bar{B}^j(u)$.
$$B^{j}(u):=\{x\in V~|~(1+\eps)^j\leq d(u,x)< (1+\eps)^{j+1}\} \textrm{ and } \bar{B}^{j}(u):=\{x\in V~|~ d(u,x)< (1+\eps)^{j+1}\}.$$
We use the following lemma from \cite{girth-icalp}. It roughly says that for each node $v$, without doing Dijkstra from $v$ to compute the balls of different radii, we can sample nodes in these balls such that they hold a particular property. This property says that the number of nodes in each level of the Dijkstra around $v$ that are close to these samples is upper bounded with high probability.

\begin{lemma}[\cite{girth-icalp}]
\label{lemma:modifiedDijkstra}
Given a graph $G$ let $M$ be the maximum edge weight of $G$ and suppose that $i\in \{1,\ldots,\log_{1+\epsilon}Mn\}$, $\beta>0$ and $0<\alpha<1$ are given. Suppose that $Q$ is a given sampled set of size $\tilde{O}(n^{\alpha})$ vertices. Let $d=\beta (1+\epsilon)^{i+1}$.  Let $V'_i=\{u\in V~|~\exists q\in Q:~d(u,q)\leq d \textrm{ and } d(q,u)\leq d\}$. In $\tilde{O}(mn^\alpha)$ time, for every $v\in V$ and every $j=\{1,\ldots, \log_{(1+\epsilon)}(Mn)\}$, one can output a sample set $R^j_{i}(v)$ of size $O(\log^2{n})$ from $\bar{Z}^j_i(v)=\bar{B}^j(v)\setminus V'_i$, where the number of vertices in ${Z}^j_i(v)=B^j(v)\setminus V'_i$ of distance at most $d$ to all vertices in $R^j_{i}(v)$ is at most $O(n^{1-\alpha})$ whp.
\end{lemma}

As a warm-up to the proofs of \cref{thm:2k+1ansc-approx} and  the improvement of it which is stated in \cref{thm:2-ansc-approx}, we first prove \cref{thm:2k+1ansc-approx} for the case of $k=1$.

\begin{theorem}
\label{thm:3approxansc}
For every $\epsilon>0$, there is a $(3+\epsilon)$-approximate algorithm for ANSC in directed graphs with edge weights in $\{1,\ldots,M\}$ that runs in $\tilde{O}(m\sqrt{n}\log{M}/\epsilon)$ time.
\end{theorem}

\begin{proof}
We are going to describe an algorithm that for each $v$ gives the following estimate for $SC(v)$:  If $(1+\epsilon)^i\le SC(v)\le (1+\epsilon)^{i+1}$, then the algorithm gives an estimate of at most $3(1+\epsilon)^2$. This gives an overal approximation factor of $3(1+\epsilon)^3\le 3(1+O(\epsilon))$ if $\epsilon<1$.

First we sample a set $Q$ of size $\tilde{O}(\sqrt{n})$ from $V(G)$ uniformly at random and perform Dijkstra from all nodes in $Q$. These Dijkstras will give us a primitive estimate on the shortest cycle passing through each node. Let $\beta=(1+\epsilon)$ and fix some $i$. Let $d=\beta(1+\epsilon)^{i+1}$ and let $V'_i=\{u\in V~|~\exists q\in Q:~d(u,q)\leq d \textrm{ and } d(u,v)\leq d\}$. Lemma \ref{lemma:modifiedDijkstra} gives us the sets $R_i^j(v)\subseteq \bar{Z}^j_i(v)=\bar{B}^j(v)\setminus V'_i$ for each $v$, such that the number of nodes in ${Z}^j_i(v)=B^j(v)\setminus V'_i$ that are at distance at most $d$ to all nodes in $R_i^j(v)$ is at most $\sqrt{n}$. Let these nodes be $S_{i}^j(v)$. So 
$$
S_i^j(v)=\{u\in Z_i^j(v)| \forall r\in R_i^j(v):d(v,r)\le d\}
$$
and $|S_i^j(v)|\le \sqrt{n}$.

Consider $v$ with $SC(v)\le d/(1+\epsilon)$. We show that each node $u\in C_v$ is either in $V'_i$, or $u\in S_i^j(v)$ for some $j$.

Suppose that $u\notin V'_i(v)$. Let $j$ be the level of $u$, so $u\in Z_i^j(v)$. Note that for each $w\in \bar{Z}_i^j(v)$, we have that $d(v,w)\le (1+\epsilon)d(v,u)$. So $d(u,w)\le d(u,v)+d(v,w)\le d(u,v)+(1+\epsilon)d(v,u)\le (1+\epsilon)SC(v)\le d$. So this is true for all $w\in R_i^j(v)$ as well, and so $u\in S_i^j(v)$.

Note that if $u\in V'_i$, then the primitive estimate for $v$ is at most $3d$. This is because there is $q$ such that $d(u,q),d(q,u)\le d$. So $d(v,q)+d(q,v)\le SC(v)+d(u,q)+d(q,u)< 3d$, and hence the Dijkstra from $q$ give an estimate of at most $3d$ for $SC(v)$.
So for each $v$ with $SC(v)\le d/(1+\epsilon)$ and primitive estimate more than $3d$, we have that $u\in S_i^j(v)$ for some $j$ for all $u\in C_v$. So the rest of the algorithm is as follows: 

For each $i\in\{1,\ldots, \log_{1+\epsilon}Mn\}$, and for each node $v$ with premitive estimate more than $3d=3\beta(1+\epsilon)^{i+1}$, we do the following ``modified" Dijkstra. We begin by placing $v$ in a Fibonacci heap with $d[v]=0$ and all other vertices with $d[\cdot]=\infty$. When a vertex $x$ is extracted, let $j$ be its level in the BFS tree from $v$: determine $j$ for which $x\in B^j(v)$. Check if $x\in S_i^j(v)$: for every $r\in R_i^j(v)$, check if $d(x,r)\le d$. If $x\notin S_i^j(v)$, ignore it. Otherwise, go through all out-edges $(x,y)$ of $x$ and if $d[y]>d[x]+w(x,y)$, update $d[y]$. We stop when the vertex extracted from the heap has $d[x]>d$. Moreover, if the extracted node has an edge to $v$, we can update the estimate of $SC(v)$ to $d[x]+w(x,v)$ if this value is smaller than the current estimate. Since we proved that all nodes in $C_v$ are in $S_i^j(v)$ for appropriate $j$'s, we get the exact value of $SC(v)$ in this case. 

So if for some $v$ we have $\beta(1+\epsilon)^i\le SC(v)<\beta(1+\epsilon)^{1+i}$, then the algorithm gives an estimate of at most $3\beta(1+\epsilon)^{2+i}$.
\end{proof}

\begin{proofof}{Theorem \ref{thm:2k+1ansc-approx}}
We show that the girth algorithm of \cite{girth-icalp} with minor modifications can be used to estimate shortest cycles, and we lay out the modifications. First, we describe the high level idea behind the girth algorithm. The algorithm starts by sampling $Q$ of size $\tilde{O}(n^{\alpha_k})$, doing in and out Dijkstra from all $q\in Q$, and defining the sets $V'_i=\{v\in V~|~\exists q\in Q:~d(v,q)\leq \beta(1+\epsilon)^{i+1} \textrm{ and } d(q,v)\leq \beta(1+\epsilon)^{i+1}\}$, where $\beta=k+O(\epsilon)$. If $i_{min}$ is the minimum $i$ where $V'_i$ is non-empty, then the algorithm has an initial estimate of $2\beta(1+\epsilon)^{i_{min}+1}$ for the girth, and has to see if the girth is smaller than this amount. The algorithm continues in $i_{min}-1$ stages, where at stage $i$ for $i<i_{min}$, the algorithm looks for cycles of length at most $2\beta(1+\epsilon)^{i+1}$.

The algorithm looks for small cycles by looking at cycles passing through each node. Initially all nodes are labeled ``on", and the algorithm does a more efficient Dijkstra (modified Dijkstra) from an on node $v$, and at the end of the Dijkstra it sets a subset of nodes ``off", and continues until there is no on node. In particular it does the modified Dijkstra up to a radius $r$ containing only on nodes, recurses on this subgraph, and sets a subset of nodes off. The recursion base sets only one node off.
The nodes labeled off mean that we know if the shortest cycle passing through them is larger that $2\beta(1+\epsilon)^{i+1}$ or not. The algorithm has following property:

(*) If $u$ is a node in a cycle $C$ with length at most $\beta(1+\epsilon)^{i+1}$ that is set off in a recursion on subgraph $H$, and none of the nodes of $C$ have been set off before, then $C$ is contained in $H$. 


The modifications to this algorithm are as follows: We will have $\log_{1+\epsilon}n$ stages instead of $i_{min}-1$ stages, one for each $i=1,\ldots,\log_{1+\epsilon}Mn$. At each stage, we want to get a good estimate for nodes $v$ with $(1+\epsilon)^i\le SC(v)< (1+\epsilon)^{i+1}$. Since we already have a good estimate for such nodes that are in $V'_i$, we set the nodes in $V'_i$ off and all the other nodes on at the beginning of stage $i$. The algorithm of stage $i$ is the same as the girth algorithm, except that we won't stop if we found a cycle smaller than $\beta(1+\epsilon)^{i+1}$. We simply update the shortest cycle estimate of ``all" the nodes of that cycle. This won't change the running time of the algorithm.

We show that if $(1+\epsilon)^i\le SC(v)< (1+\epsilon)^{i+1}$, we get a $\le (2k+1+O(\epsilon))SC(v)$ estimate for $SC(v)$ at stage $i$. 
Suppose that the first node in $C_v$ that is set off at stage $i$ is $u$. First note that if $u=v$, then either $v\in V'_i$ so we have a $2\beta(1+\epsilon)^{i+1}\le 2\beta(1+\epsilon)SC(v)=(2k+O(\epsilon))SC(v)$ estimate for it, or $v$ was set off in a recursion step, which means that by recursion we have a $(2k+1+O(\epsilon))SC(v)$ estimate for it. So suppose that $v\neq u$. If $u\notin V'_i$, then $u$ was set off in a recursion step on some subgraph $H$. By property (*), $C_v$ is contained in $H$. So the recursion gives an estimate of $(2k+1+O(\epsilon))SC(v)$ for $v$. Note that it doesn't necessarily set $v$ off, but the estimate update on $v$ in this step is good enough. If $u\in V'_i$, then it means that there is $q$ such that $d(u,q),d(q,v)\le \beta(1+\epsilon)^{i+1}$. We show that the Dijkstra from $q$ gives a good estimate for $SC(v)$: $d(q,v)+d(v,q)\le d(q,u)+d(u,v)+d(v,u)+d(v,q)\le SC(v)+2\beta(1+\epsilon)^{i+1}\le (2k+1+O(\epsilon))SC(v) $. 

\end{proofof}

\subsection{$(2+\epsilon)$-approximation algorithm for directed ANSC}
The algorithm is similar to Theorem \ref{thm:3approxansc}, with a making a few lemmas used in this algorithm more accurate. In particular, in the proof of Lemma \ref{lemma:modifiedDijkstra} Lemma \ref{lemma:setreduce-0} stated below is used. We are going to change the statement of these lemmas without actually changing their proofs. The main idea in Theorem \ref{thm:3approxansc} is being able to sample from the balls of different radii around each node without actually performing a BFS from the nodes, such that the sample set has a property. Here we want to sample from the levels (not the balls) and have a more restricted property as a result. 

\begin{lemma}[\cite{girth-icalp}]
Let $G=(V,E)$ be a directed graph with $|V|=n$ and integer edge weights in $\{1,\ldots,M\}$.
Let $S\subseteq V$ with $|S|>c \log n$ (for $c\geq 100/\log(10/9)$) and let $d$ be a positive integer.
Let $R$ be a random sample of $c\log n$ nodes of $S$ and define
$S':=\{s\in S~|~d(s,r)\leq d,~\forall r\in R\}.$
Suppose that for every $s\in S$ there are at most $0.2 |S|$ nodes $v\in V$ so that $d(s,v),d(v,s)\leq d$.
Then $|S'|\leq 0.8 |S|$.
\label{lemma:setreduce-0}
\end{lemma}

We change the last condition of Lemma \ref{lemma:setreduce-0} from supposing that for every $s\in S$, there are at most $0.2|S|$ nodes $v\in V$ so that $d(s,v),d(v,s)\le d$, to supposing that for every $s\in S$, there are at most $0.2|S|$ nodes $v\in S$ (instead of $v\in V$) so that $d(s,v),d(v,s)\le d$. So the resulting lemma is the following.

\begin{lemma}
Let $G=(V,E)$ be a directed graph with $|V|=n$ and integer edge weights in $\{1,\ldots,M\}$.
Let $S\subseteq V$ with $|S|>c \log n$ (for $c\geq 100/\log(10/9)$) and let $d$ be a positive integer.
Let $R$ be a random sample of $c\log n$ nodes of $S$ and define
$S':=\{s\in S~|~d(s,r)\leq d,~\forall r\in R\}.$
Suppose that for every $s\in S$ there are at most $0.2 |S|$ nodes $v\in S$ so that $d(s,v),d(v,s)\leq d$.
Then $|S'|\leq 0.8 |S|$.
\label{lemma:new-setreduce}
\end{lemma}

One can easily verify that by replacing $v\in V$ with $v\in S$, the proof of the lemma by \cite{girth-icalp} doesn't change and it is still true, and so we don't include the proof here. By restrticting the set of $v$'s in Lemma \ref{lemma:new-setreduce}, we show that in Lemma \ref{lemma:modifiedDijkstra} we can replace $V'_i$ by the following more restricted set: $V^j_i(v)=\{u\in B^j(v)| \exists q\in Q\cap {B}^j(v): d(u,q),d(q,u)\le d\}$. Note that $V_i^j(v)\subset V'_i$ for all $v$ and we can compute $V_i^j(v)$ for all $i,j,v$ when we do in and out Dijkstra from the nodes in $Q$.


Now we have to prove this new version of Lemma \ref{lemma:modifiedDijkstra}.

\begin{lemma}
\label{lemma:new-modifiedDijkstra}
Given a graph $G$ let $M$ be the maximum edge weight of $G$ and suppose that $i\in \{1,\ldots,\log_{1+\epsilon}Mn\}$, $\beta>0$ and $0<\alpha<1$ are given. Suppose that $Q$ is a given sampled set of size $\tilde{O}(n^{\alpha})$ vertices. Let $d=\beta (1+\epsilon)^{i+1}$.  Let $V^j_i(v)=\{u\in B^j(v)| \exists q\in Q\cap {B}^j(v): d(u,q),d(q,u)\le d\}$. In $\tilde{O}(mn^\alpha)$ time, for every $v\in V$ and every $j=\{1,\ldots, \log_{(1+\epsilon)}(Mn)\}$, one can output a sample set $R^j_{i}(v)$ of size $O(\log^2{n})$ from ${Z}^j_i(v)={B}^j(v)\setminus V_i^j(v)$, where the number of vertices in ${Z}^j_i(v)$ of distance at most $d$ to all vertices in $R^j_{i}(v)$ is at most $O(n^{1-\alpha})$ whp.
\end{lemma}

The proof of Lemma \ref{lemma:new-modifiedDijkstra} lies in the following fact: Suppose that we are able to pick a random sample $R_{i}^j(v)$ of $c\log n$ vertices from ${Z}_i^j(v)$. Then we can define $B_{i}^j(v)=\{z\in {Z}_i^j(v)~|~d(z,y)\leq d ,~\forall y\in R_i^j(v)\}$. We want to bound the size of $B_i^j(v)$.

Consider any $s\in {Z}_i^j(v)$ with at least $0.2 |{Z}_i^j(v)|$ nodes $w\in {Z}_i^j(v)$ such that $d(s,w),d(w,s)\leq d$. As ${Z}_i^j(v)\geq 10n^{1-\alpha}$ (as otherwise we would be done and the sampled vertices would work), $0.2 |{Z}_i^j(v)|\geq 2n^{1-\alpha}$, and so with high probability, for $s$ with the property above, $Q\cap {B}^j(v)$ contains some $q$ with $d(s,q),d(q,s)\leq d$, and so $s\in V^j_i(v)$. Thus with high probability, for every $s\in {Z}^j_i$, there are at most $0.2 |{Z}_i^j(v)|$
 nodes $v\in {Z}_i^j(v)$ so that $d(s,v),d(v,s)\leq d$. This means that using Lemma \ref{lemma:new-setreduce} we can conclude that $|B_i^j(v)|\le 0.8|{Z}_i^j(v)|$. If we could further sample from $B_i^j(v)$, we could shrink this set more. The rest of the proof of Lemma \ref{lemma:modifiedDijkstra} is about this sampling procedure and can be used without modification for Lemma \ref{lemma:new-modifiedDijkstra}.

Now we can prove a $(2+\epsilon)$-approximation for directed ANSC.

\begin{theorem}\label{thm:2-ansc-approx}
Given an $n$-node $m$-edge directed graph $G$ with edge weights in $\{1,\ldots,M\}$, and a constant $\epsilon>0$, there is a randomized $(2+\epsilon)$-approximate algorithm for ANSC that runs in $\tilde{O}(m\sqrt{n}\log{M})$ time with high probability.
\end{theorem}

\begin{proof} The algorithm is the same as Theorem \ref{thm:3approxansc} where we use Lemma \ref{lemma:new-modifiedDijkstra} instead of Lemma \ref{lemma:modifiedDijkstra} to produce the sets $R_i^j(v)$: Let $\beta=(1+\epsilon)$ and let $d=\beta(1+\epsilon)^{i+1}$. For each node $v$, we sample $R_i^j(v)$ from $Z_i^j(v)$ such that if $S_i^j(v)=\{u\in Z_i^j(v)|\forall r\in R_i^j(v): d(u,r)\le d\}$, then $|S_i^j(v)|\le n^{1-\alpha}$ whp.

Suppose that $(1+\epsilon)^i\le SC(v)<(1+\epsilon)^{i+1}=d/(1+\epsilon)$. Suppose that $u\in C_v$ and $u\in B^j(v)$. 
We again show that either $u\in V_i^j(v)$ or $u\in S_i^j(v)$ for some $j$ such that $u\in Z_i^j(v)$. If $u\notin V_i^j(v)$, then for all $w\in {Z}_i^j(v)$ we have $d(u,w)\le d(u,v)+d(v,w)\le d(u,v)+d(v,u)(1+\epsilon)\le SC(v)(1+\epsilon)\le d$. So for all nodes $r\in R_i^j(v)$ we have $d(u,r)\le d$ and so $u\in S_i^j(v)$. 

We show that stage $i$ of the algorithm gives a $2+O(\epsilon)$-approximation of $SC(v)$. First suppose that $C_v$ has no nodes in $V_i^j(v)$ for any $j$. Then in the modified Dijkstra from $v$, all the nodes in $C_v$ are present since for any $u\in C_v$, $u\in S_i^j(v)$ and the modified Dijkstra visits all the nodes in $S_i^j(v)$. So we visit all of the nodes in $C_v$ and we find this cycle. In this case, our estimate is $SC(v)$.  

Now suppose that there is $u\in C_v$ and some $j$ such that $u\in V_i^j(v)$. Then there is $q\in Q\cap {B}^j(v)$ such that $d(u,q),d(q,u)\le d$. So the paths $vq,qu$ and $uv$ contain a cycle of length at most $d(v,q)+d(q,u)+d(u,v)\le d(v,u)+d+d(u,v)\le 2d$. So the Dijkstra from $q$ gives an estimate of $2d\le (2+O(\epsilon))SC(v)$ for $v$. 
\end{proof}

\section{Approximation Algorithms for Undirected ANSC}
\input{ansc-under2-approx}

\subsection{Algorithms with close to linear running time}
In this section our goal is to get algorithms with $\tilde{O}(n^{1+1/k}+m)$ running time. 
We first state an approximation algorithm where the multiplicative factor has exponential dependency on $k$, and then we reduce this dependency to $O(k^2)$. Note that we need the first algorithm as we use it to approximate small cycles in the second algorithm.
\smallcycles*
\begin{proof}
Let $P$ be a $1$-fault tolerant $(2k-1)$-spanner of size $O(n^{1+1/k})$ for $G$ that can be constructed in $\tilde{O}(m)$ time by Lemma \ref{lem:fault-tolerant-const}. We transform the graph so that the degrees of vertices in $P$ are uniform. To do this, we replace $v$ and the edges attached to it by a balanced tree with  $O(d(v)/n^{1/k})$ nodes, with $v$ being the root. We assign the original edges adjacent to $v$ in $P$ to the leaves of this tree, so that each leaf gets at most $n^{1/k}$ edges. If an edge $uv$ is not in $P$, we just add it to the graph, by putting an edge between the root of the balanced tree on $v$ and the root of the balanced tree on $u$. The internal edges of the trees will have weight zero. Note that we have added $O(n)$ nodes and edges. So we can assume that in $P$ every node has degree $O(n^{1/k})$, and edges have weight zero or one. 

Let $V_i\subseteq V(G)$ be the nodes with degree between $n^{\frac{i-1}{k}}$ and $n^{\frac{i}{k}}$, for $i=1,\ldots,k$.
Note that we can assume that $P$ includes the edges incident to all the nodes in $V_1$. 

We first run the $k$-approximation algorithm for ANSC of Theorem \ref{thm:kapprox-ansc-undir} in $\tilde{O}(n^{1+2/k})$ time on $P$. We later show that this step gives a good estimation for $SC(v)$ where $v$ is a low degree node.

Let $S_i$ be a sample of size $\tilde{O}(n^{\frac{k-i}{k}})$ for $i=1,\ldots,k-1$, to hit the set of $n^{i/k}$ nearest nodes to each node. We define $p_i(v)$ for each node $v\in V(G)$ to be the closest node in $S_i$ to $v$. We find $p_i(v)$ for each $v$ by doing Dijkstra from the set $S_i$ for each $i$. Let $F(s)$ be the nodes $v$ such that $s=p_i(v)$. Let $G_s$ be the graph consisting of $P$ and the edges in $G$ adjacent to at least one node in $F(s)$. So if $v\in F(s)$, all the edges adjacent to $v$ are present in $G_s$.

For each $s\in S_i$, we do  \textsc{CycleEstimationDijkstra}$(s)$ in $G_s$ up to radius $d(s,p_{i+1}(s))-1$. Note that in expectation we see $O(n^{\frac{i+1}{k}})$ nodes, and each node has degree $O(n^{1/k})$ in $P$ so the Dijkstra from the nodes in $S_i$ take $\tilde{O}(n^{1+2/k}+m)$ time. .

Now we prove that we have a good estimate for all nodes. 

Consider a node $v$, and define $D:=(2k-1)SC(v)$. By Lemma \ref{lem:cycle-est-fault-tolerant} there is a cycle of length at most $D$ in $G_{p_j(v)}$ for all $j$, since all the edges attached to $v$ in $G$ exist in $G_{p_j(v)}$. 

First note that if $v\in V_1$, since all the edges adjacent to $v$ exist in $P$, by Lemma \ref{lem:cycle-est-fault-tolerant} there is a cycle of length at most $D$ around $v$ in $P$, and so the algorithm of Theorem \ref{thm:kapprox-ansc-undir} gives a $kD$ estimate for this cycle. 

Next if $v\in V_{k-1}$, we also have a good estimate for $SC(v)$: $p_{k-1}(v)$ is a neighbor of $v$, and we do full Dijkstra from it in $G_{p_{k-1}(v)}$. So our estimate for $v$ is $D+2\le (2k-1)SC(v)+2$.

So suppose that $v\in V_j$ for $k-1>j>1$.
This means that $p_{j-1}(v)$ must be a neighbor of $v$. Now if for all $k-1>i\ge j-1$ we have $d(v,p_{i+1}(v))\le 2d(v,p_i(v))+\floor{D/2}$, then since $k-2\ge i\ge 1$, we have $d(v,p_{k-1}(v))\le 2^{k-2}(1+D/2)-D/2$. So the estimate that we get from the  \textsc{CycleEstimationDijkstra}$(p_{k-1}(v))$ is at most $2^{k-1}(1+D/2)\le 2^{k-1}+2^{k-2}(2k-1)SC(v)$.

Now if $d(v,p_{k-1}(v))> 2^{k-2}(1+D/2)+D/2$, since $d(v,p_{j-1}(v))=1$, there is $k-1>i\ge j-1$ such that $d(v,p_{i+1}(v))\ge 2d(v,p_i(v))+\floor{D/2}+1$. Take the smallest such $i$, so that for each $j-1\le t<i$ we have $d(v,p_{t+1}(v))\le  2d(v,p_t(v))+\floor{D/2}$. So if we define $s=p_i(v)$, we have $d(v,s)\le 2^{i-j+1}(1+D/2)-D/2\le 2^{k-2}(1+D/2)-D/2$. 

Now note that we have $d(s,p_{i+1}(s))\ge d(v,p_{i+1}(s))-d(v,s)\ge d(v,p_{i+1}(v))-d(v,s)\ge d(v,s)+\floor{D/2}+1$. So the cycle of length at most $D$ around $v$ in $G_{s}$ is in the Dijkstra tree of $s$, and our estimate for it is $2d(v,s)+D$. Using the bound for $d(v,s)$, we get the estimate of $2^{k-1}+2^{k-2}(2k-1)SC(v)$.

So we get a $((2k-1)2^{k-2},2^{k-1})$-approximation algorithm in $\tilde{O}(m+n^{1+2/k})$ time. 
\end{proof}


\begin{theorem}
\label{thm:k^2-approx-ansc}
Given an $n$-node $m$-edge undirected unweighted graph $G$, and a constant integer $k\ge 2$, there is a randomized algorithm that computes a $(k^2,k^32^{k+1})$-approximation for ANSC in $\tilde{O}(n^{1+2/k}+m)$ time with high probability.
\end{theorem}

\begin{proof}
Consider the $(k,k-1)$ spanner $H_{k,k-1}$ of $G$. By Lemma \ref{lem:cycle-spanner-general}, if $SC(v)>2k(k+2)$, then there is a cycle of length at most $kSC(v)+(k+2)(k-1)$ around $v$. So we apply the Algorithm of Theorem \ref{thm:kapprox-ansc-undir} to get a $k$-approximation for ANSC on $H_{k,k-1}$ in $\tilde{O}(n^{1+1/k}\cdot n^{1/k})$. This way we get a $(k^2,k(k-1)(k+2))$-approximation for nodes $SC(v)$ for the nodes $v$ with $SC(v)>2k(l+2)$. 

Now for small cycles, we use the algorithm of Theorem \ref{thm:4k2^2k-approx} on $G$ that works in $\tilde{O}(m+n^{1+2/k})$ time. Since small cycles are of length at most $2k(k+2)$, this algorithm returns cycles of length at most $2^{k-1}(k(k+2)(2k-1)+1)$. For each $v$, we take the minimum of these two estimates as our final estimate, and so we get a $(k^2,k^32^{k+1})$-approx algorithm. 
\end{proof}

\section{Approximation Algorithm for $ST$-Shortest Paths}
\label{app:st}

We study a natural special case of $n$-PSP where the pairs $(s_i,t_i)$ come from the product set $S\times T$, where $S$ and $T$ are vertex subset. We call this special case \emph{$ST$-Shortest Paths}.
There is a 2-approximation for $n$-PSP in $\tilde{O}(m\sqrt{n})$ time, implied by the distance oracle of Agarwal \cite{DBLP:conf/esa/Agarwal14}. We show that this running time can be improved for $ST$-Shortest Paths where $|S|,|T|=O(\sqrt{n})$ (so that the total number of vertex pairs is still $O(n)$). For this problem, we provide an algorithm with running time $\tilde O(m\cdot n^{(1+\omega)/8})$. Notice that since $\omega<3$ this is polynomially faster than the $\tilde{O}(m\sqrt n)$ time algorithm that simply computes the distances from every node in $S$.

\begin{restatable}{theorem}{st}\label{thm:st}
Given an $n$-node $m$-edge undirected weighted graph $G$ and vertex pairs $(s_i,t_i)$ for $1\leq i\leq O(n)$, there is a randomized $2$-approximation algorithm for $ST$-Shortest Paths when $|S|,|T|= O(\sqrt{n})$, in $\tilde O(m\cdot n^{(1+\omega)/8})$ time with high probability.
\end{restatable}

 The idea is to use the power of sparse matrix multiplication to quickly determine  which of the truncated Dijkstras from $s_i$ and $t_i$ overlap, for all $(s_i,t_i)$ pairs at once.
We will need the following lemma on sparse rectangular matrix multiplication.

\begin{lemma}[Sparse rectangular matrix multiplication, {\cite[Theorem 2.5(i)]{DBLP:conf/compgeom/KaplanSV06}}, see also \cite{DBLP:journals/talg/YusterZ05}]
Let $A$ be an $a\times b$ matrix, and $B$ be an $b \times a$ matrix, each containing at most $c$ non-zero entries. If $c\ge a^{(\omega+1)/2}$, then computing $AB$ can be done in $c\cdot a^{(\omega-1)/2}$ time.
\label{lem:sparserect}
\end{lemma}

\noindent\begin{proofof}{\cref{thm:st}}
From every node $u \in  S\cup T$, we use BFS to compute the ball around $u$ with $r$ edges, denoted by $T_u$, where $r$ is a parameter to be determined later.

For every pair of $s\in S$ and $t\in T$, there are three possible cases:

\begin{itemize}
\item Case 1: $T_s$ and $T_t$ have no common vertices.
    
    The algorithm for this case is the same as before: at the beginning we uniformly sample $\tilde O(m/r)$ edges of the graph, and let $X$ denote the set of all the endpoints of these edges. Then, with high probability, $X$  hits all the BFS trees. We compute shortest paths from every node in $X$, in $\tilde O(m\cdot |X|)=\tilde O(m^2/r)$ total time. 
    
    For the vertex pair $(s,t)$, without loss of generality we may assume the radius of $T_s$ is smaller than that of $T_t$. Letting $v \in T_s\cap X$, we have the estimation $d(s,t)\le d(s,v)+d(v,t)\le 2\cdot d(s,v)+d(s,t)\le 2\cdot d(s,t)$.
    
    \item Case 2: $T_s$ and $T_t$ have common vertices.
    \begin{itemize}
        \item Case 2(a): $t \in T_s$ or $s \in T_t$.
        
    In this case, we can immediately obtain the exact value of $d(s,t)$ after BFS.
        \item Case 2(b): $t \notin T_s$ and $s \notin T_t$.
        
        For any $u\in T_s \cap T_t$, we must have $d(s,u)\le d(s,t)$ and $d(u,t)\le d(s,t)$.
    Hence, it suffices to find an arbitrary node $u \in T_s \cap T_t$, and then estimate the distance by $d(s,t) \le d(s,u)+d(u,t) \le  2\cdot d(s,t)$.
    
    We will use sparse rectangular matrix multiplication to solve this task for all $(s,t)\in S\times T$.  Define an $|S|\times n$ Boolean matrix $A$, where the $u$-th entry on the $s$-th row equals $1$ if and only if $u\in T_s$. Similarly, define an $n\times |T|$ Boolean matrix $B$, where the $u$-th entry on the $t$-th column equals $1$ if and only if $u\in T_t$. Let $C=AB$. Then, $C_{s,t}$ is non-zero if and only if  $T_s \cap T_t$ is non-empty. Moreover, when $C_{s,t}$ is non-empty, we can also find a witness $u$ with almost no overhead using standard sampling techniques \cite{DBLP:journals/jc/GalilM93}. 
    
    The number of non-zeros in $A$ and $B$ is at most  $|S|\cdot r = O(\sqrt{n}r)$.  
     Then, by  \cref{lem:sparserect},  
we can multiply $A$ and $B$ in $O(\sqrt{n}r \cdot  \sqrt{n}^{(\omega-1)/2}) = O(r\cdot n^{(1+ \omega)/4})$ time, provided that
$ r \ge n^{(\omega-1)/4}$.
    \end{itemize}
\end{itemize}

Hence, the total time complexity is  $\tilde O(r\cdot n^{(1+ \omega)/4} + m^{2}/r)$, which can be made $\tilde O(m\cdot n^{(1+\omega)/8})$ by choosing $r = m / n^{(1+\omega)/8}$.
\end{proofof}


\section{Applying Agarwal's Distance Oracle}
\label{sec:agrawal}
Agarwal \cite{DBLP:conf/esa/Agarwal14} designed  distance oracles with trade-offs between preprocessing time complexity and query time complexity. They can be used to solve the $n$-PSP problem by balancing the time complexity of two parts.
\begin{theorem}[Agarwal \cite{DBLP:conf/esa/Agarwal14}]
    \label{thm:aga}
    Let $k\ge 1$.  Given a weighted undirected graph with $n$ nodes and $m$ edges, one can construct a distance oracle with strecth $1 + \frac{1}{k}$,  query time $O((\alpha m/n)^k)$, and construction time $\tilde O(mn/\alpha)$, for any parameter $1\le \alpha \le n$.

    Alternatively, one can construct a distance oracle with strecth $1 + \frac{1}{k+0.5}$,  query time $O(\alpha \cdot (\alpha m/n)^k)$, and construction time $\tilde O(mn/\alpha)$, for any parameter $1\le \alpha \le n$.
\end{theorem}
\begin{corollary}
     Given a weighted undirected graph with $n$ nodes and $m$ edges, the $n$-PSP problem can be solved with $1+\frac{1}{k}$ approximation in 
     \[ \tilde O(m^{2-2/(k+1)}n^{1/(k+1)})\]
     time.

     Alternatively, it can be solved with $1+\frac{1}{k+0.5}$ approximation in
     \[ \tilde O(m^{2-3/(k+2)}n^{2/(k+2)})\]
     time.
\end{corollary}
\begin{proof}
 We build a distance oracle from \cref{thm:aga} on the input graph, and then query the distance oracle to answer all the questions. 

    For stretch $1+\frac{1}{k}$, the total time using distance oracle is (ignoring log factors)
    \begin{align*}
        \frac{mn}{\alpha} + n\cdot \left ( \frac{\alpha m}{n}\right)^k  \approx 
        m^{2-2/(k+1)}n^{1/(k+1)}
    \end{align*}
    by setting $\alpha= (n/m) \cdot (m^2/n)^{1/(k+1)}$.

    For stretch $1+\frac{1}{k+0.5}$, the total time using distance oracle is (ignoring log factors)
    \begin{align*}
        \frac{mn}{\alpha} + n\cdot \alpha \cdot \left ( \frac{\alpha m}{n}\right)^k  \approx
       m^{2-3/(k+2)}n^{2/(k+2)}
    \end{align*}
    by setting  $\alpha =(n/m)\cdot (m^3/n^2)^{1/(k+2)} $.
    
    Note that $\alpha\ll 1$ happens precisely when the claimed time bound is larger than $nm$, and in such case we  can simply solve APSP in $\tilde O(mn)$ time.
\end{proof}

%% file: ansc-under2-approx.tex
\subsection{Nearly ($1+1/(k-1)$)-approximation for ANSC in $\tilde O(m^{2-2/k}n^{1/k})$ time}

\begin{theorem}
\label{thm:unpubansc}
Given an $n$-node $m$-edge undirected unweighted graph $G$ and an integer $k\geq 2$, there is a randomized algorithm for ANSC that computes for each vertex $v$ an estimate $\hat{c}_v$ such that $SC(v)\leq \hat{c}_v\leq SC(v)+2\left\lceil\frac{SC(v)}{2(k-1)}\right\rceil$ in $\tilde O(m^{2-2/k}n^{1/k})$ time with high probability.
\end{theorem}

\begin{proof}
This algorithm is very similar to the algorithm of Dahlgaard, B{\ae}k Tejs Knudsen, and St{\"o}ckel \cite{dahlgaard2017new} for approximating the girth of a graph. There are two differences: we sample edges instead of vertices, and we use the  \textsc{CycleEstimationDijkstra} data structure.

\paragraph{Algorithm.}
Let $x$ be a parameter to be set later. First, we pick a random sample $S$ of $\Theta(\frac{m\log n}{x})$ edges. With high probability, for every vertex $v$, $S$ hits among the closest $x$ edges to $v$, where the distance from $v$ to an edge $(u,u')$ is defined as $\min\{d(v,u),d(v,u')\}$ (and ties are broken arbitrarily). Then, for all vertices $s$ such that $s$ is an endpoint of an edge in $S$, we run \textsc{CycleEstimationDijkstra}$(s)$. 

For any vertex $v$ and any positive number $r$, let  $\overline{B}(v,r)$ denote the set of edges with at least one endpoint of distance at most $r$ from $v$. For every vertex $v$, let $r(v)$ be the largest integer $r$ such that $|\overline{B}(v,r)|\leq x$.  Let $u_1, \dots, u_n$ denote all of the vertices in sorted order such that $r(u_1)\geq\dots\geq r(u_n)$. For each $i$, let $H(u_i)$ be the graph induced by $\cup_{j\leq i}u_j$. Then, for each $i$, run truncated  \textsc{CycleEstimationDijkstra}$(u_i)$ where we stop after visiting $x^{k-1}$ edges.

For all vertices $v$, our estimate $\hat{c}_v$ is the minimum estimate obtained over all executions of \textsc{CycleEstimationDijkstra}.

\paragraph{Analysis of Correctness.}
Fix $v$ and fix a shortest cycle $C_v$ through $v$. Let $r=\min_{u\in C_v}r(u)$. Since $S$ is a hitting set for the closest $x$ edges to each vertex, $S$ hits an edge with an endpoint $y$ at distance at most $r+1$ from a vertex $z\in C_v$. Thus, by \cref{lem:datastructure}, we have $\hat{c}_v\leq SC(v)+2(r+1)$, which gives the desired estimate if $r+1\leq \left\lceil  \frac{SC(v)}{2(k-1)}\right\rceil  $.
 Thus, in the following, we assume that \begin{equation}
     r+1> \left\lceil  \frac{SC(v)}{2(k-1)}\right\rceil .\label{eqn:assumption2}
 \end{equation}
 
 Let $w$ be the vertex on $C_v$ that appears last in the ordering $u_1,\dots,u_n$. That is, $r(w)=r$, and $H(w)$ contains every vertex on $C_v$. By definition, for all $u\in H(w)$, $r\leq r(u)$. Thus, for all $u\in H(w)$, $|\overline{B}(u,r)|\leq x$.

The following claim will be useful.
\begin{claim}
\label{claim:r1r2-path}
Suppose in graph $G$, there exist nonnegative integers $r_1,x_1,r_2,x_2$ such that every vertex $v$ satisfies $|\overline{B}(v,r_1)|\le x_1$ and $|\overline{B}(v,r_2)|\le x_2$. Then every vertex $v$ satisfies $|\overline{B}(v,r_1+r_2+1)|\le x_1\cdot x_2$.
\end{claim}
\begin{proof}
Given arbitrary edge $(u_1,u_2)\in \overline{B}(v,r_1+r_2+1)$, assume $d(v,u_1) \le r_1+r_2+1$, and let $(p,q)$ be the edge on an arbitrary shortest path from $v$ to $u_1$ such that $d(v,p)+1=d(v,q) = \min\{r_1+1,d(v,u_1)\}$. So $(p,q) \in \overline{B}(v,r_1)$. Then, we have $d(q,u_1) \le r_2$, so $(u_1,u_2) \in \overline{B}(q,r_2)$. We associate $(u_1,u_2)$ to vertex $q$. Then each $q$ is associated with at most $|\overline{B}(q,r_2)| \le x_2$ edges. The number of such $q$ is at most $|\overline{B}(v,r_1)| \le x_1$. So $|\overline{B}(v,r_1+r_2+1)| \le x_1\cdot x_2$.
\end{proof}
 
 By iterating \cref{claim:r1r2-path}, we have that $|\overline{B}_{H(w)}(w,(k-1)r+k-2)|\leq x^{k-1}$. By \cref{eqn:assumption2}, $(k-1)r+k-2\geq\lceil SC(v)/2\rceil$.
 Thus, by \cref{lem:datastructure}, $\hat{c}_v=SC(v)$.

\paragraph{Analysis of time complexity.} The total time for \textsc{CycleEstimationDijkstra}$(s)$ for each $s\in S$ is $\tilde{O}(|S| \cdot m) = \tilde O(m^2/x)$. The total time for truncated \textsc{CycleEstimationDijkstra} from all nodes  is at most $\tilde O(n\cdot x^{k-1})$. Setting $x = (m^2/n)^{1/k}$, the total time complexity is $\tilde{O}(n^{1/k}\cdot m^{2-2/k})$.
\end{proof}